\documentclass{svjour3}                     
\smartqed  
\usepackage{graphicx}
\usepackage{bm}
\usepackage{amsmath}
\usepackage{amssymb}
\usepackage{latexsym}
\usepackage{epsfig}
\usepackage{amsbsy}
\usepackage{array}
\usepackage{amssymb}
\usepackage{setspace}
\usepackage{algorithm}
\usepackage{algorithmic}
\numberwithin{equation}{section}
\numberwithin{theorem}{section}
\numberwithin{lemma}{section}
\numberwithin{remark}{section}
\usepackage{slashbox}

\def\sint{\ifmmode{- \!\!\!\!\!\! \int}
    \else{\hbox{$- \!\!\!\! \int \ $}}\fi}

\usepackage{mathptmx}      
\begin{document}

\title{Numerical algorithms for the forward and backward fractional Feynman-Kac equations}


\author{Weihua Deng$^{*,1}$, Minghua Chen$^{1}$, Eli Barkai$^{2}$
}


\institute{
$^{*}$Corresponding author. E-mail: dengwh@lzu.edu.cn            \\
$^{1}$School of Mathematics and Statistics,
Lanzhou University, Lanzhou 730000, P. R. China \\
$^{2}$Department of Physics and Advanced Materials and Nanotechnology Institute, Bar-Ilan University, Ramat Gan 52900, Israel.
E-mail: Eli.Barkai@biu.ac.il}

\date{Received: date / Accepted: date}

\maketitle

\begin{abstract}
The Feynman-Kac equations are a type of partial differential equations describing the distribution of functionals of diffusive motion. The probability density function (PDF) of Brownian functionals satisfies the Feynman-Kac formula, being a Schr\"{o}dinger equation in imaginary time. The functionals of no-Brownian motion, or anomalous diffusion, follow the fractional Feynman-Kac equation [J. Stat. Phys. {\bf 141}, 1071-1092, 2010], where the fractional substantial derivative is involved. Based on recently developed discretized schemes for fractional substantial derivatives [arXiv:1310.3086], this paper focuses on providing algorithms for numerically solving the forward and backward fractional Feynman-Kac equations; since the fractional substantial derivative is non-local time-space coupled operator, new challenges are introduced comparing with the general fractional derivative. Two ways (finite difference and finite element) of discretizing the space derivative  are considered. For the backward fractional Feynman-Kac equation, the numerical stability and convergence of the algorithms with first order accuracy are theoretically discussed; and the optimal estimates are obtained. For all the provided schemes, including the first order and high order ones, of both forward and backward Feynman-Kac equations, extensive numerical experiments are performed to show their effectiveness.
\keywords{fractional Feynman-Kac equation \and fractional substantial derivative \and optimal convergent order \and
 numerical stability and  convergence \and numerical inversion of Laplace transforms}
\end{abstract}


\section{Introduction}
Letting $x(t)$ be a trajectory of a Brownian particle and $U(x)$ be a prescribed function, the Brownian functional can be defined as $A=\int_0^t U[x(\tau)]d\tau$ \cite{Majumdar:05}, which has many physical applications. In 1949, inspiring by Feynman's path integrals Kac derives a Schr\"{o}dinger-like equation for the distribution of the functionals of diffusive motion \cite{Kac:49}. With the rapid development on the study of non-Brownian motion, or anomalous diffusion \cite{Mandelbrot:68,Metzler:00}, the functionals of anomalous diffusion naturally attract the interests of physicists.  In particular, Carmi, Turgeman, and Barkai derive the forward and backward fractional Feynman-Kac equations for describing the distribution of the functionals of anomalous diffusion
\cite{Carmi:11,Carmi:10,Turgeman:09}, which involves the fractional substantial derivative \cite{Friedrich:06}. Being the same form of Brownian functional, the functional of anomalous diffusion can also be defined as
\begin{equation}\label{1.1}
  A=\int_0^t U[x(\tau)]d\tau,
\end{equation}
where $x(t)$ is a trajectory of non-Brownian particle; and there are a lot of different choice to prescribe $U(x)$. For example, we can take $U(x)=1$ in a given domain and to be zero otherwise, which characterizes the time spent by a particle in the domain; this functional can be used in kinetic studies of chemical reactions that take place exclusively in the domain \cite{Agmon:84,Carmi:10}. For inhomogeneous disorder dispersive systems, the motion of the particles is non-Brownian, and $U(x)$ is taken as $x$ or $x^2$ \cite{Carmi:10}.

In recent decades, the numerical methods for fractional partial differential equations (PDEs) are well developed, including finite difference methods \cite{Chen:12,Chen:09,Meerschaert:06,Sun:06,Zhou:13}, finite element \cite{Deng:08,Ervin:05,Jiang:11}, spectral method \cite{Li:12,Li:09}, etc.  However, it seems that there are no published works for numerically solving fractional PDEs with fractional substantial derivative. Fractional substantial derivative is a non-local time-space coupled operator; discretizing it and numerically solving the corresponding equations undoubtedly introduce some new difficulties comparing with the fractional derivative.  We detailedly discuss the properties and effectively numerical discretizations of the fractional substantial derivatives in \cite{Chen:13}. This paper focuses on numerically solving the forward and backward fractional Feynman-Kac equations with the fractional substantial derivative being discretized by the ways given in \cite{Chen:13} and the classical spatial derivative is treated by finite difference and finite element method, respectively. For the backward Feynman-Kac equation, we theoretically prove the numerical stability and convergence of its first order scheme. For all the proposed schemes, including the first order and high order ones, of both forward and backward fractional Feynman-Kac equations, the extensive numerical experiments are performed to show their effectiveness.

The definitions of fractional substantial calculus are given as follows \cite{Chen:13}.
\begin{definition}\label{definition1.1}
Let $ \nu>0$, $\rho$ be a  constant, and $P(t)$ be piecewise continuous on $(0,\infty)$ and integrable on any finite subinterval  $[0,\infty)$. Then
the fractional substantial  integral  of $P(t)$ of order $\nu$ is defined as
\begin{equation*}
{^s\!}I_t^\nu P(t)=\frac{1}{\Gamma(\nu)}\int_{0}^t{\left(t-\tau\right)^{\nu-1}}e^{-\rho U(x)(t-\tau)}{P(\tau)}d\tau, ~~~~t>0,
\end{equation*}
where $U(x)$ is a prescribed  function in (\ref{1.1}).
\end{definition}
\begin{definition}\label{definition1.2}
Let $ \mu>0$, $\rho$ be a  constant, and $P(t)$ be (m-1)-times continuously differentiable on $(a,\infty)$ and its $m$-times derivative be integrable on any finite subinterval of  $[a,\infty)$, where $m$ is the smallest integer that exceeds $\mu$. Then
the fractional substantial  derivative  of $P(t)$ of order $\mu$ is defined as
\begin{equation*}
{^s\!}D_t^\mu P(t)={^s\!}D_t^m[{^s\!}I_t^{m-\mu} P(t)],
\end{equation*}
where
\begin{equation*}
 {^s\!} D_t^m=\left(\frac{\partial}{\partial t}+\rho  U(x)\right)^m.
\end{equation*}
\end{definition}
The forward and backward fractional Feynman-Kac equation derived in \cite{Carmi:11,Carmi:10,Turgeman:09} are
\begin{equation}\label{1.2}
\frac{\partial}{\partial t}P(x,\rho,t)=\kappa_\alpha \, \frac{\partial^2}{\partial x^2} {^s\!}D_t^{1-\alpha} P(x,\rho,t)-\rho U(x)P(x,\rho,t), \end{equation}
and
\begin{equation}\label{1.3}
\frac{\partial}{\partial t}P(x,\rho,t)=\kappa_\alpha \, {^s\!}D_t^{1-\alpha} \frac{\partial^2}{\partial x^2}P(x,\rho,t)-\rho U(x)P(x,\rho,t),
\end{equation}
where $P(x,\rho ,t):=\int_0^\infty P(x,A,t)e^{-\rho A}dA$; for (\ref{1.2}), $P(x,A,t)$ denotes the joint probability density function (PDF) of finding the particle on $(x,A)$ at time $t$; while for (\ref{1.3}), $P(x,A,t)$ is the joint PDF of finding the particle on $A$ at time $t$ with the initial position of the particle at $x$;  the functional $A$  is defined as (\ref{1.1}); the diffusion coefficient $\kappa_\alpha $ is a positive constant and $\alpha \in (0,1)$; when $U(x)=0$,  both (\ref{1.2}) and (\ref{1.3}) reduces to the celebrated fractional Fokker-Planck equation \cite{Barkai:01,Metzler:00}. In fact, from the definition of fractional substantial derivative, Eq. (\ref{1.3}) can be rewritten as
\begin{equation}\label{1.4}
{^s\!}D_t P(x,\rho,t)= {^s\!}D_t^{1-\alpha} \left[ \kappa_\alpha \frac{\partial^2}{\partial x^2}P(x,\rho,t)\right];
\end{equation}
then we can further get its equivalent form (see the Appendix)
\begin{equation}\label{1.5}
\begin{split}
{^s_c}{D}_t^\alpha P(x,t)
&={^s\!}D_t^\alpha [P(x,t)-e^{-\rho U(x) t}P(x,0)]\\
&={^s\!}D_t^\alpha P(x,t)-\frac{t^{-\alpha}e^{-\rho U(x) t}}{\Gamma(1-\alpha)}P(x,0)
  = \kappa_\alpha \frac{\partial^2}{\partial x^2} P(x,t),
\end{split}
\end{equation}
here and in the following $P(x,\rho,t)$ is replaced by $P(x,t)$ since $\rho$ is taken as a fixed constant. For (\ref{1.2}), from the definition of fractional substantial derivative, we can also recast it as
\begin{equation}\label{1.6}
 {^s\!}D_tP(x,t)=\kappa_\alpha \frac{\partial^2}{\partial x^2}\,{^s\!}D_t^{1-\alpha}P(x,t);
\end{equation}
but it should be noted that the two operators $\frac{\partial^2}{\partial x^2}$ and $D_t^{1-\alpha}$ do not commute. 

The outline of this paper is as follows. In Section 2, for (\ref{1.2}) and (\ref{1.3}) we derive the numerical schemes with finite difference method to discretize the space derivative; and theoretically prove that the first order time discretization scheme is unconditionally stable and convergent for (\ref{1.3}). In Section 3, for (\ref{1.3}) the time semi-discretized and full discretized schemes of finite element method are provided; stability and convergence of the schemes are rigourously established; moreover, the optimal convergent rate is obtained. To confirm the theoretical results and show the effectiveness of the first order and high order schemes, the extensive numerical results are provided in Section 4. We conclude the paper with some remarks in the last section.

\section{Finite difference for fractional Feynman-Kac equation}
In this section we focuses on deriving the difference schemes for the backward fractional Feynman-Kac equation (\ref{1.4})  and theoretically prove that the provided first order time discretization scheme of (\ref{1.4}) is unconditionally stable and convergent; the difference schemes for the forward fractional Feynman-Kac equation (\ref{1.2}) are given as a remark. 

Letting $T>0$, $\Omega=(0,l)$, rewriting (\ref{1.5}) and making it subject to the given initial and boundary conditions, we have
\begin{equation}\label{2.1}
 {^s_c}{D}_t^\alpha P(x,t) ={^s\!}D_t^\alpha [P(x,t)-e^{-\rho U(x) t}P(x,0)]= \kappa_\alpha \frac{\partial^2}{\partial x^2} P(x,t),
  ~~0<t \leq T,~~x \in \Omega,
\end{equation}
with  initial and boundary conditions
\begin{equation}\label{2.2}
\begin{split}
&P(x,0)=\phi(x),~~x\in \Omega,\\
&P(0,t)=\psi_1(t),~~P(l,t)=\psi_2(t), ~~0<t \leq T.
\end{split}
\end{equation}

\subsection{Derivation of the difference scheme}
Let the mesh points $x_i=ih$ for $i=0,1,\ldots,M$, and $t_n=n\tau$,
$n=0,1,\ldots,N$,   where $h=l/M$ and $\tau=T/N$ are the uniform space stepsize and time steplength, respectively.
Denote $P_i^n$ as the numerical approximation to $P(x_i,t_n)$.
To approximate  (\ref{2.1}), we utilize the second order central difference formula for the spatial derivative; that is
\begin{equation*}
\frac{\partial^2 P(x,t)}{\partial x^2}\Big|_{(x_i,t_n)}=\frac{P(x_{i+1},t_n)  -2P(x_i,t_n)+P(x_{i-1},t_n)}{h^2}+\mathcal{O}\left(h^2 \right).
\end{equation*}
From (3.8) of \cite{Chen:13}, we know that the fractional substantial derivative has $q$-th order approximations, i.e.,
\begin{equation}\label{2.3}
{^s\!}D_t^\alpha P(x,t)|_{(x_i,t_n)}=\tau^{-\alpha}\sum_{k=0}^{n}{d}_{i,k}^{q,\alpha}P(x_i,t_{n-k})+\mathcal{O}(\tau^q), ~~q=1,2,3,4,5,
\end{equation}
with
$${d}_{i,k}^{q,\alpha}=e^{-\rho U_i k \tau}{l}_k^{q,\alpha},~~U_i=U(x_i),~~q=1,2,3,4,5,$$
where ${l}_k^{1,\alpha}$, ${l}_k^{2,\alpha}$,    ${l}_k^{3,\alpha}$, ${l}_k^{4,\alpha}$ and  ${l}_k^{5,\alpha}$
are defined by (2.2),  (2.4), (2.6), (2.8) and (2.10) in \cite{Chen:1313}, respectively. In the following, we do the detailed theoretical analysis for  the first order time  discretization scheme of (\ref{2.1}). For the simplification, we denote $d_{i,k}^{1,\alpha}$ by $d_{i,k}^\alpha$; then
\begin{equation}\label{2.4}
\begin{split}
&{^s\!}D_t^\alpha P(x,t)|_{(x_i,t_n)}=\tau^{-\alpha}\sum_{k=0}^{n}{d_{i,k}^\alpha}P(x_i,t_{n-k})+\mathcal{O}(\tau); \\
& {^s\!}D_t^\alpha [e^{-\rho U(x) t}P(x,0)]_{(x_i,t_n)}=\tau^{-\alpha}\sum_{k=0}^{n}{d_{i,k}^\alpha}e^{-\rho U_i (n-k)\tau}P(x_i,0)+\mathcal{O}(\tau),
\end{split}
\end{equation}
where the coefficients
\begin{equation}\label{2.5}
  {d_{i,k}^\alpha}=e^{-\rho U_i k \tau}g_k,~~ g_k=(-1)^k\left ( \begin{matrix} \alpha \\ k\end{matrix} \right ),
\end{equation}
with
\begin{equation*}
  g_0=1, ~~~~g_k=\left(1-\frac{\alpha+1}{k}\right)g_{k-1},~~k \geq 1.
\end{equation*}

Then Eq.  (\ref{2.1}) can be rewritten as
\begin{equation}\label{2.6}
\begin{split}
&\tau^{-\alpha}\sum_{k=0}^{n}{d_{i,k}^\alpha}P(x_i,t_{n-k})-\tau^{-\alpha}\sum_{k=0}^{n}{d_{i,k}^\alpha}e^{-\rho U_i (n-k)\tau}P(x_i,0)\\
&\quad=\kappa_\alpha \frac{P(x_{i+1},t_n)  -2P(x_i,t_n)+P(x_{i-1},t_n)}{h^2}+r_i^n,
\end{split}
\end{equation}
with
\begin{equation}\label{2.7}
  |r_i^n| \leq C_P(\tau+h^2),
\end{equation}
where $C_P$ is a constant depending only on $P$.

Multiplying (\ref{2.6}) by $\tau^\alpha$, we have the following equation
\begin{equation}\label{2.8}
\begin{split}
&\sum_{k=0}^{n-1}{d_{i,k}^\alpha}P(x_i,t_{n-k})-\sum_{k=0}^{n-1}{d_{i,k}^\alpha}e^{-\rho U_i (n-k)\tau}P(x_i,0)\\
&\quad=\kappa_\alpha\tau^{\alpha}\frac{P(x_{i+1},t_n)  -2P(x_i,t_n)+P(x_{i-1},t_n)}{h^2}+R_i^n,
\end{split}
\end{equation}
with
\begin{equation}\label{2.9}
|R_i^n|=|\tau^{\alpha}r_i^n| \leq C_P \tau^{\alpha}(\tau+h^2).
\end{equation}
From (\ref{2.5}) and (\ref{2.8}),  the resulting discretization of (\ref{2.1}) can be rewritten as
\begin{equation}\label{2.10}
 P_i^n-\frac{\kappa_\alpha \tau^{\alpha}}{h^2} \left(P_{i+1}^n-2P_{i}^n+P_{i-1}^n\right)
=\sum_{k=0}^{n-1}{d_{i,k}^\alpha}e^{-\rho U_i  (n-k)\tau}P_i^{0}-\sum_{k=1}^{n-1}{d_{i,k}^\alpha}P_i^{n-k},~~n\geq 1;
\end{equation}
or
\begin{equation}\label{2.11}
 P_i^n-\frac{\kappa_\alpha \tau^{\alpha}}{h^2} \left(P_{i+1}^n-2P_{i}^n+P_{i-1}^n\right)
=\sum_{k=0}^{n-1}e^{-\rho U_i  n\tau}  g_kP_i^{0}-\sum_{k=1}^{n-1}e^{-\rho U_i k \tau}g_kP_i^{n-k},~~n\geq 1,
\end{equation}
with  $i=1,2,\ldots,M-1$. It is worthwhile to noting that the second term
on the right hand side of (\ref{2.10}) or (\ref{2.11}), respectively,  automatically vanishes when $n =1$.

\begin{remark}\label{remark2.1}
If we utilize the $q$-th order approximation of (\ref{2.3})  to  discretize  the time fractional substantial derivative of (\ref{2.1}),
the resulting discretization of (\ref{2.1}) is
\begin{equation}\label{2.12}
{d}_{i,0}^{q,\alpha} P_i^n-\frac{\kappa_\alpha \tau^{\alpha}}{h^2} \left(P_{i+1}^n-2P_{i}^n+P_{i-1}^n\right)
=\sum_{k=0}^{n-1}{d}_{i,k}^{q,\alpha}e^{-\rho U_i  (n-k)\tau}P_i^{0}-\sum_{k=1}^{n-1}{d}_{i,k}^{q,\alpha}P_i^{n-k},
\end{equation}
which gives a local truncation error of $\mathcal{O}\left(\tau^q+h^2 \right), q=2,3,4,5.$
\end{remark}

\begin{remark}\label{remark2.2}



Using the second order central difference formula for the spatial derivative leads to
\begin{equation*}
\begin{split}
&\frac{\partial^2}{\partial x^2}{^s\!}D_t^{1-\alpha}P(x_i,t_n)\\
&\quad =\frac{{^s\!}D_t^{1-\alpha}P(x_{i+1},t_n) -2\,{^s\!}D_t^{1-\alpha} P(x_{i},t_n)+{^s\!}D_t^{1-\alpha}P(x_{i-1},t_n)}{h^2}+\mathcal{O}\left(h^2 \right).
\end{split}
\end{equation*}
Further applying (\ref{2.3}) to  approximate the time fractional substantial
 derivatives, then we get the discretization schemes of  (\ref{1.6}):
 \begin{equation}\label{2.14}
 \begin{split}
 &d_{i,0}^{q,1}P_i^n  -\frac{\kappa_\alpha \tau^{\alpha}}{h^2} (d_{i+1,0}^{q,1-\alpha}P_{i+1}^n -2d_{i,0}^{q,1-\alpha}P_i^n+ d_{i-1,0}^{q,1-\alpha}P_{i-1}^n)\\
&\quad  =-\sum_{k=1}^{n}d_{i,k}^{q,1}P_i^{n-k}+ \frac{\kappa_\alpha \tau^{\alpha}}{h^2}
 \sum_{k=1}^{n}{ (d_{i+1,k}^{q,1-\alpha}}P_{i+1}^{n-k} -2d_{i,k}^{q,1-\alpha}P_i^{n-k}+ d_{i-1,k}^{q,1-\alpha}P_{i-1}^{n-k}),
 \end{split}
\end{equation}
with the local truncation error $\mathcal{O}\left(\tau^q+h^2 \right),\,q=1,2,3,4,5.$
\end{remark}
\subsection{Stability and convergence}
In this subsection, we prove that the scheme (\ref{2.10}) is unconditionally stable and convergent in discrete $L^2$ norm and $L^\infty$ norm under the assumption that $0 \leq \rho U_i \leq \eta$. First, we introduce some relevant notations and properties of discretized inner product given in \cite{Hu:99}. Denote  $u^n=\{u_i^n| 0 \leq i \leq M, n \geq 0 \}$
and $v^n=\{v_i^n| 0 \leq i \leq M, n \geq 0 \}$, which are grid functions. And
\begin{equation*}
\begin{split}
&(u_i^n)_x=(u_{i+1}^n - u_{i}^n)/h, ~~~~(u_i^n)_{\overline{x}}=(u_{i}^n - u_{i-1}^n)/h; \\
&(u^n,v^n)=\sum_{i=1}^{M-1}u_i^nv_i^nh, ~~~~~||u^n||=(u^n,u^n)^{1/2}; \\
&(u^n,v^n]=\sum_{i=1}^{M}u_i^nv_i^nh, ~~~~~~~||u^n]|=(u^n,u^n]^{1/2}. \\
\end{split}
\end{equation*}
In particular, if $u_0^n=0$ and $u_M^n=0$, there exists
\begin{equation}\label{2.15}
\begin{split}
(u^n,(v^n_{\overline{x}})_x)=-(u_{\overline{x}}^n,v_{\overline{x}}^n] ~~{\rm and}~~||u^n||^2 \leq \frac{l^2}{8}||u^n_{\overline{x}}]|^2,
\end{split}
\end{equation}
where $l$ means the one appeared in $\Omega=(0,l).$

\begin{lemma}\label{lemma2.1}
The coefficients $g_k$ defined in (\ref{2.5}) satisfy
\begin{equation}\label{2.16}
\begin{split}
 & g_0=1; ~~~~ g_k<0,~~ (k\geq 1); ~~~~\sum_{k=0}^{n-1}g_k>0;  ~~~~\sum_{k=0}^{\infty}g_k=0;
  \end{split}
\end{equation}
and
\begin{equation}\label{2.17}
\begin{split}
& \frac{1}{n^\alpha \Gamma(1-\alpha)}< \sum_{k=0}^{n-1}g_k=-\sum_{k=n}^{\infty}g_k \leq \frac{1}{n^\alpha}, ~~{\rm for}~~n\geq 1.
  \end{split}
\end{equation}
\end{lemma}
\begin{proof}
From \cite[p.\,208]{Podlubny:99}, it is easy to get (\ref{2.16}). Next we prove (\ref{2.17}).
Denoting $v_n=-n^\alpha \sum_{k=n}^\infty g_k=n^\alpha \sum_{k=0}^{n-1} g_k, \,n\geq 1 $, according to  \cite{Chen:09}, there exists
\begin{equation}\label{2.18}
 v_{n+1}< v_n, ~~{\rm i.e.,}~~ \sum_{k=0}^n g_k<\frac{n^\alpha}{(n+1)^\alpha} \sum_{k=0}^{n-1} g_k,~~{\rm for}~~n\geq 1,
\end{equation}
and
$$\frac{1}{n^\alpha \Gamma(1-\alpha)}< \sum_{k=0}^{n-1}g_k=-\sum_{k=n}^{\infty}g_k,~~{\rm for}~~n\geq 1.$$
Next we prove the following inequality by mathematical induction
\begin{equation}\label{2.19}
  \sum_{k=0}^{n-1}g_k=-\sum_{k=n}^{\infty}g_k \leq \frac{1}{n^\alpha},~~{\rm for}~~n\geq 1.
\end{equation}
It is obvious that (\ref{2.19}) holds when $n=1$ or $n=2$. Supposing that
\begin{equation*}
  \sum_{i=0}^{s-1}g_i=-\sum_{i=s}^{\infty}g_i \leq \frac{1}{s^\alpha}, ~~s=1,2,\ldots, n-1,
\end{equation*}
and using (\ref{2.18}), we obtain
$$ \sum_{k=0}^{n-1}g_k<\frac{(n-1)^\alpha}{n^\alpha} \sum_{k=0}^{n-2} g_k\leq \frac{(n-1)^\alpha}{n^\alpha} \frac{1}{(n-1)^\alpha}=\frac{1}{n^\alpha},
~~{\rm for}~~n\geq 2. $$
Then the desired inequality (\ref{2.17}) holds.
\end{proof}

\begin{theorem}\label{theorem2.1}
When $0 \leq \rho U_i \leq \eta$, the difference scheme (\ref{2.11})  is unconditionally stable.
\end{theorem}
\begin{proof}
Let $\widetilde{P_i}^n$ be the approximate solution of $P_i^n$,
which is the exact solution of the scheme (\ref{2.11}). Taking $e_i^n=\widetilde{P_i}^n- P_i^n$, $i=1,2,\ldots,M-1$, then from (\ref{2.11})
we get the following perturbation equation
\begin{equation}\label{2.20}
e_i^n-\tau^{\alpha}\kappa_\alpha\frac{e_{i+1}^n-2e_{i}^n+e_{i-1}^n}{h^2} =
\sum_{k=0}^{n-1}e^{-\rho U_i  n\tau}  g_ke_i^{0}-\sum_{k=1}^{n-1}e^{-\rho U_i k \tau}g_ke_i^{n-k},
\end{equation}
with $e_0^n=e_M^n=0$. Multiplying (\ref{2.20}) by $he_i^n$ and summing up for $i$ from $1$ to $M-1$, then
\begin{equation}\label{2.21}
\begin{split}
&h\sum_{i=1}^{M-1}(e_i^n)^2-\tau^{\alpha}\kappa_\alpha h \sum_{i=1}^{M-1}\frac{e_{i+1}^n-2e_{i}^n+e_{i-1}^n}{h^2}e_i^n \\
&\quad=h\sum_{i=1}^{M-1}\sum_{k=0}^{n-1}e^{-\rho U_i  n\tau}  g_ke_i^{0}e_i^n-h\sum_{i=1}^{M-1}\sum_{k=1}^{n-1}e^{-\rho U_i k \tau}g_ke_i^{n-k}e_i^n.
\end{split}
\end{equation}
Next we estimate  (\ref{2.21}).
Firstly, we have
\begin{equation}\label{2.22}
h\sum_{i=1}^{M-1}(e_i^n)^2=||e^n||^2,
\end{equation}
and from (\ref{2.15}), it leads to
\begin{equation}\label{2.23}
\begin{split}
&-\tau^{\alpha}\kappa_\alpha h \sum_{i=1}^{M-1}\frac{e_{i+1}^n-2e_{i}^n+e_{i-1}^n}{h^2}e_i^n\\
&\qquad =-\tau^{\alpha}\kappa_\alpha h \sum_{i=1}^{M-1}((e{_i^n})_{\overline{x}})_xe_i^n
=-\tau^{\alpha}\kappa_\alpha   (e^n,(e^n_{\overline{x}})_x]\\
&\qquad =\tau^{\alpha}\kappa_\alpha   (e^n_{\overline{x}},e^n_{\overline{x}} ]
=\tau^{\alpha}\kappa_\alpha   ||e^n_{\overline{x}}]|^2
\geq \frac{8\tau^{\alpha}\kappa_\alpha }{l^2} ||e^n||^2 \geq 0.
\end{split}
\end{equation}
Since $e^{-\rho U_i n\tau} \in [e^{-\eta  T},1]$ and   from (\ref{2.5}) and (\ref{2.16}),  we obtain
\begin{equation}\label{2.24}
 g_k \leq d_{i,k}^\alpha= e^{-\rho U_i k \tau}g_k \leq e^{-\eta  T} g_k<0,~~  k\geq 1,
\end{equation}
and
\begin{equation}\label{2.25}
0< e^{-\rho U_i  n\tau}  \sum_{k=0}^{n-1} g_k=\sum_{k=0}^{n-1}e^{-\rho U_i  n\tau}  g_k  \leq  \sum_{k=0}^{n-1}g_k, ~~n\geq 1.
\end{equation}
Therefore, according to (\ref{2.25}) and  (\ref{2.24}), we obtain
\begin{equation}\label{2.26}
\begin{split}
h\sum_{i=1}^{M-1}\sum_{k=0}^{n-1}e^{-\rho U_i  n\tau}  g_ke_i^{0}e_i^n
&\leq   h\sum_{i=1}^{M-1}\sum_{k=0}^{n-1}e^{-\rho U_i  n\tau} g_k\frac{(e_i^{0})^2 +(e_i^{n})^2}{2}\\
&\leq   \frac{1}{2}\sum_{k=0}^{n-1} g_k  \left(||e^0||^2+||e^n||^2 \right),
\end{split}
\end{equation}
and
\begin{equation}\label{2.27}
\begin{split}
-h\sum_{i=1}^{M-1}\sum_{k=1}^{n-1}e^{-\rho U_i k \tau}g_ke_i^{n-k}e_i^n
& \leq -h\sum_{i=1}^{M-1}\sum_{k=1}^{n-1}e^{-\rho U_i k \tau}g_k \frac{(e_i^{n-k})^2 +(e_i^{n})^2}{2}\\
&\leq  - \frac{1}{2}\sum_{k=1}^{n-1} g_k  \left(||e^{n-k}||^2+||e^n||^2 \right).
\end{split}
\end{equation}
From (\ref{2.21}-\ref{2.27}), there exists
\begin{equation}\label{2.28}
\begin{split}
||e^n||^2  &\leq \frac{1}{2}\sum_{k=0}^{n-1} g_k  \left(||e^0||^2+||e^n||^2 \right)- \frac{1}{2}\sum_{k=1}^{n-1} g_k  \left(||e^{n-k}||^2+||e^n||^2 \right)\\
&=\left(1+ \frac{1}{2}\sum_{k=1}^{n-1} g_k\right)  ||e^0||^2- \frac{1}{2}\sum_{k=1}^{n-1} g_k ||e^{n-k}||^2.
\end{split}
\end{equation}
Next we prove that $||e^n||^2 \leq ||e^0||^2$ by mathematical induction. For $n=1$, (\ref{2.28}) holds obviously. Supposing
\begin{equation*}
||e^s||^2 \leq ||e^0||^2, ~~{\rm for}~~s=1,2,\ldots,n-1,
\end{equation*}
and using (\ref{2.28}), then we get
\begin{equation*}
\begin{split}
||e^n||^2
&\leq\left(1+ \frac{1}{2}\sum_{k=1}^{n-1} g_k\right)  ||e^0||^2- \frac{1}{2}\sum_{k=1}^{n-1} g_k ||e^{n-k}||^2\\
&\leq\left(1+ \frac{1}{2}\sum_{k=1}^{n-1} g_k\right)  ||e^0||^2- \frac{1}{2}\sum_{k=1}^{n-1} g_k ||e^{0}||^2
=||e^0||^2.
\end{split}
\end{equation*}
Hence, the proof is complete.
\end{proof}

\begin{lemma}\label{lemma2.2}
Let $R \geq 0$; $\varepsilon^k \geq 0$,~$k=0,1,\ldots,N$ and satisfy
\begin{equation}\label{2.29}
  \varepsilon^n\leq  - \sum_{k=1}^{n-1} g_k  \varepsilon^{n-k}+R,~~n\geq 1,
\end{equation}
then we have the following estimates:

(a) when $0<\alpha<1$,
\begin{equation}\label{2.30}
\varepsilon^n \leq   \left(\sum_{k=0}^{n-1}g_k\right)^{-1}R \leq  n^\alpha \Gamma(1-\alpha)R;
\end{equation}

(b) when $\alpha \rightarrow1$,
\begin{equation}\label{2.31}
\varepsilon^n \leq     nR.
\end{equation}
\end{lemma}
\begin{proof}
It is worth to noting that the first term on the right hand side of (\ref{2.29}) automatically vanishes when $n =1$.

(1) Case $0<\alpha<1$: We  prove the following estimate by the  mathematical induction,
 $$\varepsilon^n \leq   \left(\sum_{k=0}^{n-1}g_k\right)^{-1}R .$$
Eq. (\ref{2.29}) holds obviously for $n=1$. Supposing  that
$$\varepsilon^s \leq   \left(\sum_{i=0}^{s-1}g_i\right)^{-1}R,~~s=1,2,\ldots,n-1,$$
then form (\ref{2.29}) we have
\begin{equation*}
\begin{split}
\varepsilon^n
&\leq  - \sum_{k=1}^{n-1} g_k  \varepsilon^{n-k}+R
\leq - \sum_{k=1}^{n-1} g_k  \left(\sum_{i=0}^{n-k-1}g_i\right)^{-1}R+R
\leq - \sum_{k=1}^{n-1} g_k  \left(\sum_{i=0}^{n-1}g_i\right)^{-1}R+R\\
&\leq \left(1- \sum_{k=0}^{n-1}g_k \right)   \left(\sum_{i=0}^{n-1}g_i\right)^{-1}R+R
 \leq   \left(\sum_{i=0}^{n-1}g_i\right)^{-1}R.
\end{split}
\end{equation*}
According to  (\ref{2.17}) and the above inequality, it leads to
$$\varepsilon^n \leq   \left(\sum_{k=0}^{n-1}g_k\right)^{-1}R \leq  n^\alpha \Gamma(1-\alpha)R. $$

(2) Now we consider the case $\alpha\rightarrow 1$. Since $\Gamma(1-\alpha)\rightarrow \infty$ as $\alpha \rightarrow 1$ in  the estimate (\ref{2.30}). Therefore, we need to look for an estimate of other form.
We  prove the following estimate by the  mathematical induction:
 $$\varepsilon^n \leq   nR .$$
Eq. (\ref{2.29}) holds obviously for $n=1$. Supposing  that
$$\varepsilon^s \leq   sR,~~s=1,2,\ldots,n-1,$$
thus, from (\ref{2.29}) we get
$$  \varepsilon^n\leq  - \sum_{k=1}^{n-1} g_k  \varepsilon^{n-k}+R
\leq  - \sum_{k=1}^{n-1} g_k  (n-k)R+R
\leq  - \sum_{k=1}^{n-1} g_k  (n-1)R+R
\leq   (n-1)R+R=nR.
$$
\end{proof}

\begin{theorem}\label{theorem2.2}
 Let $P_i^n$ be the approximate solution of $P(x_i,t_n)$ computed by the difference scheme (\ref{2.11}) with the assumption $0 \leq \rho U_i \leq \eta$.  Then
$$
||P(x_i,t_n)-P_i^n|| \leq C_P \Gamma(1-\alpha)l^{1/2} T^\alpha (\tau+h^2),~~0<\alpha<1,
$$
where $C_P$ is defined by  (\ref{2.7}) and  $(x_i,t_n)\in (0,l) \times (0,T]$, $i=1,2,\ldots,M-1;~ n=1,2,\ldots,N$.
\end{theorem}
\begin{proof}
Similar to the proof of  \cite{Chen:09}, let $P(x_i,t_n)$ be the exact solution of (\ref{2.1}) at the mesh point $(x_i,t_n)$,
and $\varepsilon_i^n=P(x_i,t_n)-P_i^n$.
Subtracting (\ref{2.8}) from (\ref{2.11}) and using $\varepsilon_i^0=0$, we obtain
\begin{equation}\label{2.32}
 \varepsilon_i^n-\frac{\kappa_\alpha \tau^{\alpha}}{h^2} \left(\varepsilon_{i+1}^n-2\varepsilon_{i}^n+\varepsilon_{i-1}^n\right)
=-\sum_{k=1}^{n-1}e^{-\rho U_i k \tau}g_k \varepsilon_i^{n-k}+R_i^n,~~n\geq 1,
\end{equation}
where  $R_i^n $ is defined by (\ref{2.9}).

Multiplying (\ref{2.32}) by $h\varepsilon_i^n$ and summing up for $i$ from $1$ to $M-1$, there exists
\begin{equation}\label{2.33}
\begin{split}
&h\sum_{i=1}^{M-1}(\varepsilon_i^n)^2-\tau^{\alpha}\kappa_\alpha h \sum_{i=1}^{M-1}\frac{\varepsilon_{i+1}^n-2\varepsilon_{i}^n+\varepsilon_{i-1}^n}{h^2}\varepsilon_i^n\\
&\quad =-h\sum_{i=1}^{M-1}\sum_{k=1}^{n-1}e^{-\rho U_i k \tau}g_k\varepsilon_i^{n-k}\varepsilon_i^n
+ h\sum_{i=1}^{M-1}R_i^n\varepsilon_i^n.
\end{split}
\end{equation}
It follows from the proof of Theorem \ref{theorem2.1} that
\begin{equation}\label{2.34}
\begin{split}
&h\sum_{i=1}^{M-1}(\varepsilon_i^n)^2=||\varepsilon^n||^2; \\
&-\tau^{\alpha}\kappa_\alpha h \sum_{i=1}^{M-1}\frac{\varepsilon_{i+1}^n-2\varepsilon_{i}^n+\varepsilon_{i-1}^n}{h^2}\varepsilon_i^n  \geq 0;\\
&-h\sum_{i=1}^{M-1}\sum_{k=1}^{n-1}e^{-\rho U_i k \tau}g_k\varepsilon_i^{n-k}\varepsilon_i^n
\leq  - \frac{1}{2}\sum_{k=1}^{n-1} g_k  \left(||\varepsilon^{n-k}||^2+||\varepsilon^n||^2 \right).
\end{split}
\end{equation}
According to (\ref{2.9}), (\ref{2.7})  and (\ref{2.17}), we obtain \cite{Chen:09}
\begin{equation}\label{2.35}
\begin{split}
h\sum_{i=1}^{M-1}\left| R_i^n\varepsilon_i^n \right|
&= \tau^\alpha h\sum_{i=1}^{M-1} \left|r_i^n\varepsilon_i^n \right|
\leq\tau^\alpha h\sum_{i=1}^{M-1}\left[\frac{\tau^\alpha}{2\sum_{k=0}^{n-1}g_k}(r_i^n)^2+\frac{\sum_{k=0}^{n-1}g_k}{2\tau^\alpha}(\varepsilon_i^n)^2\right]\\
&=\frac{\tau^{2\alpha} h}{2\sum_{k=0}^{n-1}g_k}\sum_{i=1}^{M-1}(r_i^n)^2+\frac{1}{2}\sum_{k=0}^{n-1}g_k||\varepsilon^n||^2\\
&\leq \frac{\tau^{2\alpha}n^\alpha \Gamma(1-\alpha) }{2}h(M-1) C_P^2(\tau+h^2)^2+\frac{1}{2}\sum_{k=0}^{n-1}g_k||\varepsilon^n||^2\\
&\leq \frac{T^\alpha \Gamma(1-\alpha)lC_P^2 }{2}\tau^{\alpha} (\tau+h^2)^2+\frac{1}{2}\sum_{k=0}^{n-1}g_k||\varepsilon^n||^2\\
&=   \frac{C_2 }{2}\tau^{\alpha} (\tau+h^2)^2+\frac{1}{2}\sum_{k=0}^{n-1}g_k||\varepsilon^n||^2,
\end{split}
\end{equation}
where $(x_i,t_n)\in (0,l) \times (0,T]$, and
\begin{equation}\label{2.36}
C_2=lC_P^2 \Gamma(1-\alpha) T^\alpha.
\end{equation}
According to (\ref{2.34}) and (\ref{2.35}), there exists
\begin{equation}\label{2.37}
\begin{split}
||\varepsilon^n||^2
& \leq  - \frac{1}{2}\sum_{k=1}^{n-1} g_k  \left(||\varepsilon^{n-k}||^2+||\varepsilon^n||^2 \right)
  + \frac{C_2 }{2}\tau^{\alpha} (\tau+h^2)^2+\frac{1}{2}\sum_{k=0}^{n-1}g_k||\varepsilon^n||^2\\
& =  - \frac{1}{2}\sum_{k=1}^{n-1} g_k  ||\varepsilon^{n-k}||^2
  + \frac{C_2 }{2}\tau^{\alpha} (\tau+h^2)^2+\frac{1}{2}||\varepsilon^n||^2,
\end{split}
\end{equation}
that is
\begin{equation}\label{2.38}
\begin{split}
||\varepsilon^n||^2 \leq   - \sum_{k=1}^{n-1} g_k  ||\varepsilon^{n-k}||^2 + C_2 \tau^{\alpha} (\tau+h^2)^2.
\end{split}
\end{equation}
According to (\ref{2.36})-(\ref{2.38}) and Lemma \ref{lemma2.2}, we have
 \begin{equation*}
||\varepsilon^n||^2 \leq   n^\alpha \Gamma(1-\alpha) C_2\tau^{\alpha} (\tau+h^2)^2
=lC_P^2 \Gamma(1-\alpha) T^\alpha \Gamma(1-\alpha) T^\alpha(\tau+h^2)^2.
\end{equation*}
Hence
\begin{equation*}
\begin{split}
||P(x_i,t_n)-P_i^n||=||\varepsilon_i^n|| \leq  C_P \Gamma(1-\alpha)l^{1/2} T^\alpha (\tau+h^2).
\end{split}
\end{equation*}
\end{proof}

Besides the discrete $L^2$ norm, the unconditional stability and convergence can also be obtained in $L^\infty$ norm. In the following theorem, we present the convergent result in $L^\infty$ norm; because of the similar proof, we omit the proof of unconditional stability in $L^\infty$ norm.

\begin{theorem}\label{theorem2.3}
Let $P_i^n$ be the approximate solution of $P(x_i,t_n)$ computed by use of
the difference scheme (\ref{2.11}) with the assumption $0 \leq \rho U_i \leq \eta$.  Then the error estimates are
$$
||P(x_i,t_n)-P_i^n||_{\infty} \leq  C_P\Gamma(1-\alpha)T^\alpha (\tau+h^2),~~{\rm for}~~0<\alpha <1;
$$
and
$$
||P(x_i,t_n)-P_i^n||_{\infty} \leq C_PT \tau^{\alpha-1}(\tau+h^2),~~{\rm for}~~\alpha \rightarrow 1,
$$
where $C_P$ is defined by  (\ref{2.7}) and  $(x_i,t_n)\in (0,l) \times (0,T]$, $i=1,2,\ldots,M-1;~ n=1,2,\ldots,N$.
\end{theorem}
\begin{proof}
Let $P(x_i,t_n)$ be the exact solution of (\ref{2.1}) at the mesh point $(x_i,t_n)$, and denote $\varepsilon_i^n=P(x_i,t_n)-P_i^n$,  $\varepsilon^n=[\varepsilon_0^n,\varepsilon_1^n,\ldots, \varepsilon_{M}^n]$.
Subtracting (\ref{2.8}) from (\ref{2.10}) and using $\varepsilon_i^0=0$, we obtain
\begin{equation*}
 \varepsilon_i^n-\frac{\kappa_\alpha \tau^{\alpha}}{h^2} \left(\varepsilon_{i+1}^n-2\varepsilon_{i}^n+\varepsilon_{i-1}^n\right)
=-\sum_{k=1}^{n-1}{d_{i,k}^\alpha}\varepsilon_i^{n-k}+R_i^n,~~n\geq 1,
\end{equation*}
where  $R_i^n $ is defined by (\ref{2.9}).
Assume that  $$|\varepsilon_{i_0}^{n}|=||\varepsilon^n||_{\infty}=\max \limits _{0\leq i \leq M}|\varepsilon_i^n|$$
and $R_{\max}=\max \limits _{0\leq i \leq M,0\leq n \leq N}|R_i^n|$. Then we  have the following estimates
\begin{equation}\label{2.39}
\begin{split}
||\varepsilon^n||_{\infty}=|\varepsilon_{i_0}^n|
&\leq |\varepsilon_{i_0}^n|
   -\frac{\kappa_\alpha \tau^{\alpha}}{h^2} \left( |\varepsilon_{i_0+1}^n| -2 |\varepsilon_{i_0}^n|+ |\varepsilon_{i_0-1}^n|\right)\\
& =\left(1+2\frac{\kappa_\alpha \tau^{\alpha}}{h^2}\right)|\varepsilon_{i_0}^n|
   -\frac{\kappa_\alpha \tau^{\alpha}}{h^2} \left( |\varepsilon_{i_0+1}^n| + |\varepsilon_{i_0-1}^n|\right)\\
&\leq \left |\varepsilon_{i_0}^n
   -\frac{\kappa_\alpha \tau^{\alpha}}{h^2} \left(\varepsilon_{i_0+1}^n-2 \varepsilon_{i_0}^n + \varepsilon_{i_0-1}^n\right)\right | \\
&=\left |-\sum_{k=1}^{n-1}{d_{i_0,k}^\alpha}\varepsilon_{i_0}^{n-k}+ R_{i_0}^n\right |\\
&\leq -\sum_{k=1}^{n-1}{d_{i_0,k}^\alpha}|| \varepsilon^{n-k}||_{\infty}+ R_{\max}.
\end{split}
\end{equation}
Form (\ref{2.24}) and (\ref{2.39}), there exists
\begin{equation*}
\begin{split}
||\varepsilon^n||_{\infty}
\leq -\sum_{k=1}^{n-1}g_k|| \varepsilon^{n-k}||_{\infty}+ R_{\max}.
\end{split}
\end{equation*}
Hence, using Lemma \ref{lemma2.2}, it leads to
$$||P(x_i,t_n)-P_i^n||_{\infty} \leq  C_P\Gamma(1-\alpha)T^\alpha (\tau+h^2),~~{\rm for}~~0<\alpha <1;$$
and
$$||P(x_i,t_n)-P_i^n||_{\infty} \leq C_PT \tau^{\alpha-1}(\tau+h^2),~~{\rm for}~~\alpha \rightarrow 1.$$
\end{proof}
\begin{remark}\label{remark2.3}
When $\rho$ is an imaginary number, i.e., $\rho=\nu+i\omega$; similar to the proof of Theorem \ref{theorem2.3} but with the assumption $0 \leq \nu U_i \leq \eta$, the same results on numerical stability and convergence can be obtained.
\end{remark}

\section{Finite element method for fractional Feynman-Kac equation}
The proposed method is based on a finite difference scheme on time and Galerkin finite element in space for (\ref{1.3}).
This section is devoted to the stability analysis of the time-stepping scheme and the detailed error analysis of semidiscretization on time
and of full discretization.  In particular, the optimal convergent order is obtained.

\subsection{Variational formulation and finite element approximation for fractional Feynman-Kac equation}

Let $T>0$, $\Omega=(0,l)$,  and $t_n=n\tau$,
$n=0,1,\ldots,N$,   where $\tau=\frac{T}{N}$ is  the time steplength.
Rewriting (\ref{1.5}), and making it subject to the given initial and boundary conditions, we have
\begin{equation}\label{3.1}
 {^s_c}{D}_t^\alpha P(x,t) ={^s\!}D_t^\alpha [P(x,t)-e^{-\rho U(x) t}P(x,0)]= \kappa_\alpha \frac{\partial^2}{\partial x^2} P(x,t),
  ~~0<t \leq T,~~x \in \Omega,
\end{equation}
with the initial and boundary conditions
\begin{equation}\label{3.2}
\begin{split}
&P(x,0)=\phi(x),~~x\in \Omega,\\
&P(0,t)=P(l,t)=0, ~~0<t \leq T.
\end{split}
\end{equation}

Using the first order approximation of (\ref{2.3})  to discretize the time fractional derivative of (\ref{3.1}), denoting $d_{i,k}^{1,\alpha}$ as  $d_{k}^\alpha$, and taking $U(x)=\sigma$ being a constant,
then we obtain
\begin{equation}\label{3.3}
\begin{split}
{^s\!}D_t^\alpha[P(x,t)]_{t=t_n}
&=\tau^{-\alpha}\sum_{k=0}^{n}{d_{k}^\alpha}P(x,t_{n-k})+ {\widetilde r}_n^{(1)}(x); \\
{^s\!}D_t^\alpha [e^{-\rho U(x)  t}P(x,0)]_{t=t_n}
& = {^s\!}D_t^\alpha [e^{-\rho \sigma  t}P(x,0)]_{t=t_n} \\
&=\tau^{-\alpha}\sum_{k=0}^{n}{d_{k}^\alpha}e^{-\rho \sigma  (n-k)\tau}P(x,0)+r_n^{(1)}(x)\\
&=\tau^{-\alpha}\sum_{k=0}^{n}e^{-\rho \sigma  n\tau}g_kP(x,0)+{r}_n^{(1)}(x),
\end{split}
\end{equation}
with
\begin{equation}\label{3.4}
  |r_n^{(1)}(x)| \leq \widetilde{C}_P\tau ~~{\rm and }~~  |\widetilde{r}_n^{(1)}(x)| \leq \widetilde{C}_P\tau.
\end{equation}
Here $\widetilde{C}_P$ is a constant depending only on $P$, and the coefficients
\begin{equation}\label{3.5}
  {d_{k}^\alpha}=e^{-\rho \sigma  k \tau}g_k,~~ g_k=(-1)^k\left ( \begin{matrix} \alpha \\ k\end{matrix} \right ).
\end{equation}
Denoting $P^n(x)$ as an approximation of $P(x,t_n)$, then we get the following time discrete scheme of (\ref{3.1}):
\begin{equation}\label{3.6}
\begin{split}
P^n(x)-\kappa_\alpha \tau^{\alpha} \Delta P^n(x)
&=\sum_{k=0}^{n-1}{d_{k}^\alpha}e^{-\rho \sigma \, (n-k)\tau}P^{0}(x) - \sum_{k=1}^{n-1}{d_{k}^\alpha}P^{n-k}(x)\\
&=\sum_{k=0}^{n-1} e^{-\rho \sigma  n\tau}g_kP^{0}(x) - \sum_{k=1}^{n-1} e^{-\rho \sigma  k \tau}g_kP^{n-k}(x).
\end{split}
\end{equation}
For the simplification, we use $P^n$ to denote $P^n(x)$.
Then the variational formulation of (\ref{3.6}) subject to the boundary condition reads as follows: find $P^n \in H_0^1(\Omega)$ such that
\begin{equation}\label{3.7}
\begin{split}
&(P^n,q)-\kappa_\alpha \tau^{\alpha} (\Delta P^n,q) \\
&\quad =\sum_{k=0}^{n-1} e^{-\rho \sigma  n\tau}g_k(P^{0},q) - \sum_{k=1}^{n-1} e^{-\rho \sigma  k \tau}g_k(P^{n-k},q),~~ \forall q\in H_0^1(\Omega),
\end{split}
\end{equation}
with the initial and boundary conditions
\begin{equation}\label{3.8}
\begin{split}
&P^0(x)=\phi(x),~~x\in \Omega,\\
&P^n(0)=P^n(l)=0, ~~n \geq  1.
\end{split}
\end{equation}

\subsection{Stability analysis and error estimates for the semidiscrete scheme}

\begin{theorem}\label{theorem3.1}
The weak semidiscrete scheme  (\ref{3.7}) with $\rho\sigma$ being positive real number is unconditionally stable in the sense that for all $\tau>0$,  it holds that
$$||P^n||_{\widetilde{H}_0^1(\Omega)} \leq ||P^0||_{L^2},~~n=1,2,\ldots,N,$$
where $||P^n||_{\widetilde{H}_0^1(\Omega)}=\left(||P^n||^2_{L^2} +\kappa_\alpha \tau^{\alpha} ||\nabla P^n||^2_{L^2}\right)^{1/2}$.
\end{theorem}
\begin{proof}
Taking $q=P^n$  and from (\ref{3.7}), we obtain
$$(P^n,P^n)+\kappa_\alpha \tau^{\alpha} (\nabla P^n,\nabla P^n)
=\sum_{k=0}^{n-1} e^{-\rho \sigma  n\tau}g_k(P^{0},P^n) - \sum_{k=1}^{n-1} e^{-\rho \sigma  k \tau}g_k(P^{n-k},P^n).$$
Since  $e^{-\rho \sigma  n\tau} \in [e^{-\rho \sigma T},1]$, then from (\ref{3.5}) and  (\ref{2.16})  we obtain
\begin{equation}\label{3.9}
g_k \leq d_{k}^\alpha= e^{-\rho \sigma  k \tau}g_k <0,~~  k\geq 1,
\end{equation}
and
\begin{equation}\label{3.10}
0< e^{-\rho \sigma   n\tau} \sum_{k=0}^{n-1} g_k=\sum_{k=0}^{n-1}e^{-\rho \sigma   n\tau}  g_k  \leq \sum_{k=0}^{n-1}g_k, ~~n\geq 1.
\end{equation}
Then using Schwartz inequality, we have
\begin{equation}\label{3.11}
\begin{split}
||P^n||^2_{\widetilde{H}_0^1(\Omega)}
&=\sum_{k=0}^{n-1} e^{-\rho \sigma  n\tau}g_k(P^{0},P^n) - \sum_{k=1}^{n-1} e^{-\rho \sigma  k \tau}g_k(P^{n-k},P^n)\\
&\leq \sum_{k=0}^{n-1} e^{-\rho \sigma  n\tau}g_k||P^{0}||_{L^2}||P^{n}||_{L^2}
- \sum_{k=1}^{n-1} e^{-\rho \sigma  k \tau}g_k||P^{n-k}||_{L^2}||P^{n}||_{L^2}.
\end{split}
\end{equation}
According to (\ref{3.9})-(\ref{3.11}), there exists
\begin{equation}\label{3.12}
\begin{split}
||P^n||_{\widetilde{H}_0^1(\Omega)}
&\leq \sum_{k=0}^{n-1} g_k||P^{0}||_{L^2}- \sum_{k=1}^{n-1} e^{-\rho \sigma  k \tau}g_k||P^{n-k}||_{L^2}\\
&\leq \sum_{k=0}^{n-1} g_k||P^{0}||_{L^2}- \sum_{k=1}^{n-1}g_k||P^{n-k}||_{L^2}.
\end{split}
\end{equation}
Next we  prove $||P^n||_{\widetilde{H}_0^1(\Omega)} \leq ||P^0||_{L^2}$. The inequality (\ref{3.12}) holds obviously when $n=1$. Supposing
\begin{equation*}
||P^s||_{\widetilde{H}_0^1(\Omega)} \leq ||P^0||_{L^2}, ~~{\rm for}~~s=1,2,\ldots,n-1,
\end{equation*}
then from (\ref{3.12}), we obtain
$$||P^n||_{\widetilde{H}_0^1(\Omega)}
\leq \sum_{k=0}^{n-1} g_k||P^{0}||_{L^2}- \sum_{k=1}^{n-1}g_k||P^{n-k}||_{L^2}
\leq \sum_{k=0}^{n-1} g_k||P^{0}||_{L^2}- \sum_{k=1}^{n-1}g_k||P^{0}||_{L^2}=||P^{0}||_{L^2}.$$
The proof is complete.
\end{proof}

\begin{theorem}\label{theorem3.2}
Let $P$ be the exact solution of (\ref{3.1})-(\ref{3.2}), and $\{P^n\}_{n=0}^N$ be the time discrete  solution of (\ref{3.7})
with the initial and boundary conditions (\ref{3.8}) under the assumption $\rho\sigma>0$. Then we have the following error estimates:
$$
||P(t_n)-P^n||_{\widetilde{H}_0^1(\Omega)}\leq \widetilde{C}_P \Gamma(1-\alpha) T^\alpha   \tau,~~{\rm for}~~0<\alpha <1;
$$
and
$$
||P(t_n)-P^n||_{\widetilde{H}_0^1(\Omega)}\leq \widetilde{C}_P  T   \tau^\alpha,~~\alpha \rightarrow 1,
$$
where $\widetilde{C}_P$ is defined by  (\ref{3.4}) and  $(x,t_n)\in (0,l) \times (0,T]$, $ n=1,2,\ldots,N$.
\end{theorem}

\begin{proof}
Define $e^n=P(x,t_n)-P^n(x)$. Using $e^0=0$ and  (\ref{3.1}), (\ref{3.7})  and (\ref{3.4}), there exists
\begin{equation}\label{3.13}
\begin{split}
(e^n,q)+ \kappa_\alpha \tau^{\alpha} (\nabla e^n,\nabla q)
= - \sum_{k=1}^{n-1} e^{-\rho \sigma  k \tau}g_k(e^{n-k},q)+(R^n,q),~~ \forall q\in H_0^1(\Omega),
\end{split}
\end{equation}
where
\begin{equation}\label{3.14}
  ||R^n||_{L_2} \leq \widetilde{C}_P\tau^{1+\alpha}.
\end{equation}
Taking $q=e^n$ in (\ref{3.13}) and from (\ref{3.10}), we obtain
$$
||e^n||^2_{\widetilde{H}_0^1(\Omega)} \leq  - \sum_{k=1}^{n-1} e^{-\rho \sigma  k \tau}g_k ||e^{n-k}||_{L^2}||e^{n}||_{L^2}+||R^{n}||_{L^2}||e^{n}||_{L^2}.
$$
Then from (\ref{3.9}) and (\ref{3.14}),  it leads to
\begin{equation}\label{3.15}
||e^n||_{\widetilde{H}_0^1(\Omega)} \leq  - \sum_{k=1}^{n-1} d^\alpha_{k} ||e^{n-k}||_{L^2}+ \widetilde{C}_P\tau^{1+\alpha}.
\end{equation}

(\uppercase\expandafter{\romannumeral1}) First consider the case $0 < \alpha < 1 $. We start by proving the following estimate:
\begin{equation}\label{3.16}
||e^n||_{\widetilde{H}_0^1(\Omega)} \leq     \left(\sum_{k=0}^{n-1}d^\alpha_{k} \right)^{-1}\widetilde{C}_P \tau^{1+\alpha},~~n\geq 1.
\end{equation}
The inequality (\ref{3.15}) holds obviously when $n=1$. Supposing  that
$$
||e^s||_{\widetilde{H}_0^1(\Omega)} \leq  \left(\sum_{j=0}^{s-1}d^\alpha_{j}\right)^{-1}   \widetilde{C}_P \tau^{1+\alpha},~~s=1,2,\ldots,n-1,
$$
then from (\ref{3.15}) we obtain
\begin{equation*}
\begin{split}
||e^n||_{\widetilde{H}_0^1(\Omega)}
&\leq  - \sum_{k=1}^{n-1} d^\alpha_{k} ||e^{n-k}||_{L^2}+ \widetilde{C}_P\tau^{1+\alpha}\\
&\leq  - \sum_{k=1}^{n-1} d^\alpha_{k}  \left(\sum_{j=0}^{n-k-1}d^\alpha_{j}\right)^{-1}  \widetilde{C}_P\tau^{1+\alpha} + \widetilde{C}_P\tau^{1+\alpha}\\
&\leq \left(1 - \sum_{k=0}^{n-1} d^\alpha_{k} \right)
      \left(\sum_{j=0}^{n-1}d^\alpha_{j}\right)^{-1} \widetilde{C}_P\tau^{1+\alpha}+\widetilde{C}_P\tau^{1+\alpha}\\
&=  \left(\sum_{j=0}^{n-1}d^\alpha_{j}\right)^{-1}\widetilde{ C}_P  \tau^{1+\alpha}.
\end{split}
\end{equation*}
According to (\ref{3.9}) and (\ref{2.17}), there exists
\begin{equation*}
\begin{split}
 &\sum_{k=0}^{n-1}d^\alpha_{k} \geq \sum_{k=0}^{n-1}g_k=-\sum_{k=n}^{\infty}g_k>\frac{1}{n^\alpha \Gamma(1-\alpha)}>0, ~~~~n\geq 1.
  \end{split}
\end{equation*}
Then  using (\ref{3.16}) and the above inequality, we get
\begin{equation}\label{3.17}
||e^n||_{\widetilde{H}_0^1(\Omega)} \leq    \left(\sum_{k=0}^{n-1}d^\alpha_{k} \right)^{-1} \widetilde{C}_P\tau^{1+\alpha}
\leq   \widetilde{C}_P  n^\alpha \Gamma(1-\alpha)  \tau^{1+\alpha}
= \widetilde{C}_P \Gamma(1-\alpha) T^\alpha   \tau.
\end{equation}

(\uppercase\expandafter{\romannumeral2}) Now we consider the case $\alpha\rightarrow 1$. Since $\Gamma(1-\alpha)\rightarrow \infty$ as $\alpha \rightarrow 1$ in  the estimate (\ref{3.17}). Therefore, we need to look for an estimate of other form.
We prove the following estimate by the  mathematical induction:
\begin{equation}\label{3.18}
||e^n||_{\widetilde{H}_0^1(\Omega)}\leq \widetilde{C}_Pn\tau^{1+\alpha},~~n\geq 1.
\end{equation}
It is obvious that (\ref{3.15}) holds when $n=1$. Denoting that
\begin{equation*}
||e^s||_{\widetilde{H}_0^1(\Omega)}\leq \widetilde{C}_Ps\tau^{1+\alpha},~~s= 1,2,\ldots,n-1,
\end{equation*}
and from (\ref{3.15}), it leads to
\begin{equation*}
\begin{split}
||e^n||_{\widetilde{H}_0^1(\Omega)}
&\leq  - \sum_{k=1}^{n-1} d^\alpha_{k} ||e^{n-k}||_{L^2}+ \widetilde{C}_P\tau^{1+\alpha}
 \leq  - \sum_{k=1}^{n-1} d^\alpha_{k}  \widetilde{C}_P(n-k)\tau^{1+\alpha} + \widetilde{C}_P\tau^{1+\alpha}\\
&\leq  - \sum_{k=1}^{n-1} d^\alpha_{k}  \widetilde{C}_P(n-1)\tau^{1+\alpha} + \widetilde{C}_P\tau^{1+\alpha} \leq C_Pn\tau^{1+\alpha}.
\end{split}
\end{equation*}
Hence
$$||e^n||_{\widetilde{H}_0^1(\Omega)}\leq \widetilde{C}_PT \tau^{\alpha},~~{\rm as}~~\alpha\rightarrow 1.$$
\end{proof}

\subsection{Finite element approximation and error estimates for full discretization}
Denote $S_h$ as the piecewise polynomials of degree at most $r-1$ on mesh $\{x_i\}$, and define elliptic or Ritz projection $R_h$ from $H_0^1(\Omega)$
into $S_h$ by the orthogonal relation:
\begin{equation*}
  (\nabla R_hv, \nabla \chi)=  (\nabla v, \nabla \chi), ~~\forall \chi \in S_h, ~~{\rm for}~~v \in H_0^1.
\end{equation*}
Then we have the well-known approximation property \cite{Thomee:06}:
\begin{equation}\label{3.19}
  ||R_hv-v||_{L^2}+h||\nabla (R_hv-v)||_{L^2} \leq Ch^s||v||_s, ~~{\rm for}~~v \in H^s \cap H_0^1, 1\leq s \leq r.
\end{equation}
Letting
\begin{equation}\label{3.20}
\begin{split}
&\tau^{-\alpha}\sum_{k=0}^{n}{d_{k}^\alpha}P(x,t_{n-k})=\tau^{-\alpha}\sum_{k=0}^{n} {d_{k}^\alpha} R_hP(x,t_{n-k})+ r_n^{(2)}(x),
\end{split}
\end{equation}
then combining (\ref{3.3}) and  (\ref{3.20}), we obtain
\begin{equation}\label{3.21}
\begin{split}
{^s\!}D_t^\alpha P(x,t_n)=\tau^{-\alpha}\sum_{k=0}^{n}d_{k}^\alpha R_hP(x,t_{n-k})+r_n(x),
\end{split}
\end{equation}
with
\begin{equation}\label{3.22}
r_n(x)=r_n^{(1)}(x)+ r_n^{(2)}(x).
\end{equation}

Now we give the finite element approximation of (\ref{3.7}):
find $P_h^n \in S_h$ such that
\begin{equation}\label{3.23}
\begin{split}
&(P_h^n,q_h)-\kappa_\alpha \tau^{\alpha} (\Delta P_h^n,q_h)\\
&\quad =\sum_{k=0}^{n-1} e^{-\rho \sigma  n\tau}g_k(P_h^{0},q_h) - \sum_{k=1}^{n-1} e^{-\rho \sigma  k \tau}g_k(P_h^{n-k},q_h),~~ \forall q_h\in S_h(\Omega).
\end{split}
\end{equation}

\begin{lemma}\label{lemma3.1}
The coefficients $d_{k}^\alpha$ defined in (\ref{3.5}) with $\rho\sigma>0$ satisfy
\begin{equation*}
\frac{1}{n^\alpha \Gamma(1-\alpha)}< \sum_{k=0}^{n-1}g_k \leq \sum_{k=0}^{n-1}d_{k}^\alpha \leq \frac{1+(\rho \sigma  T)^\alpha}{n^\alpha},
~~{\rm for}~~n\geq 1.
\end{equation*}
\end{lemma}
\begin{proof}
From \cite{Chen:13}, we know that
$$(1-\frac{\zeta}{e^{\rho \sigma  \tau}})^\alpha=\sum_{k=0}^\infty d_{k}^\alpha\zeta^k.$$
Taking $\zeta=1$, then there exists
$$(1-e^{-\rho \sigma  \tau})^\alpha=\sum_{k=0}^\infty d_{k}^\alpha,$$
which leads to
\begin{equation}\label{3.24}
  \sum_{k=0}^{n-1} d_{k}^\alpha=(1-e^{-\rho \sigma \tau})^\alpha-\sum_{k=n}^\infty d_{k}^\alpha.
\end{equation}
Using (\ref{3.9}) and Lemma \ref{lemma2.1}, we obtain
$$\frac{1}{n^\alpha \Gamma(1-\alpha)}< \sum_{k=0}^{n-1}g_k \leq \sum_{k=0}^{n-1}d_{k}^\alpha,~~n \geq 1,$$
and
\begin{equation}\label{3.25}
-\sum_{k=n}^{\infty}d_{k}^\alpha \leq -\sum_{k=n}^{\infty}g_k  \leq \frac{1}{n^\alpha}, ~~n \geq 1.
\end{equation}
Next we prove
$$\sum_{k=0}^{n-1}d_{k}^\alpha \leq \frac{1+(\rho \sigma  T)^\alpha}{n^\alpha}.$$
According to (\ref{3.24}) and (\ref{3.25}), there exists
$$
\sum_{k=0}^{n-1} d_{k}^\alpha=(1-e^{-\rho \sigma \tau})^\alpha-\sum_{k=n}^\infty d_{k}^\alpha
\leq (1-e^{-\rho \sigma \tau})^\alpha +\frac{1}{n^\alpha}\leq \frac{(\rho \sigma  T)^\alpha+1}{n^\alpha},
$$
since
$$
(1-e^{-\rho \sigma \tau})^\alpha \leq (\rho \sigma \tau)^\alpha \leq \left(\rho \sigma  \frac{T}{n}\right)^\alpha
=\frac{(\rho \sigma  T)^\alpha}{n^\alpha}.
$$
\end{proof}

\begin{lemma}\label{Lemma3.3}
The truncation error $r_n(x)$ defined by (\ref{3.22}) is bounded by
\begin{equation*}
||r_n(x)||_{L^2} \leq \overline{C}_P \left(\tau+\tau^{-\alpha}\frac{(\rho \sigma  T)^\alpha+1}{n^\alpha}h^r\right),~~n=1,2,\ldots,N,
\end{equation*}
where $\overline{C}_P$ is a constant depending only on $P$.
\end{lemma}
\begin{proof}
Here, $r_n^{(1)}(x)$ is given in (\ref{3.4}),
      $$||r_n^{(1)}(x)||_{L^2} \leq \widetilde{C}_P\tau.$$
From (\ref{3.19}), there exists
\begin{equation}\label{3.26}
  ||R_hP(x,t_{n-k})-P(x,t_{n-k})||_{L^2} \leq \widehat{C}_Ph^r,
\end{equation}
where $\widehat{C}_P$ is a constant depending only on $P$.

Then, using (\ref{3.20}), (\ref{3.26}), and Lemma \ref{lemma3.1} leads to
\begin{equation*}
\begin{split}
||r_n^{(2)}(x)||_{L^2}
&\leq \tau^{-\alpha}\sum_{k=0}^{n}{d_{k}^\alpha}   ||R_hP(x,t_{n-k})-P(x,t_{n-k})||_{L^2} \\
&\leq   \tau^{-\alpha}\widehat{C}_Ph^r\sum_{k=0}^{n}{d_{k}^\alpha}  \leq \tau^{-\alpha}\widehat{C}_Ph^r\frac{(\rho \sigma  T)^\alpha+1}{n^\alpha}.
\end{split}
\end{equation*}
Hence, from (\ref{3.22}) we obtain
\begin{equation*}
\begin{split}
||r_n(x)||_{L^2}
&\leq ||r_n^{(1)}(x)||_{L^2}+ ||r_n^{(2)}(x)||_{L^2}
\leq \widetilde{C}_P\tau+\widehat{C}_P\tau^{-\alpha}\frac{(\rho \sigma  T)^\alpha+1}{n^\alpha}h^r\\
&\leq \overline{C}_P \left(\tau+\tau^{-\alpha}\frac{(\rho \sigma  T)^\alpha+1}{n^\alpha}h^r\right),
\end{split}
\end{equation*}
with $\overline{C}_p=\max(\widetilde{C}_P,\widehat{C}_P).$

\end{proof}

\begin{theorem}\label{theorem3.3}
Let $P$ be the exact solution of (\ref{3.1})-(\ref{3.2}), and $\{P_h^n\}_{n=0}^N$ be the solution of the full discretization scheme  (\ref{3.23})
with the initial condition $P_h^0=R_h\phi$ under the assumption $\rho\sigma>0$. If $P \in  H^r(\Omega) \cap H_0^1(\Omega)$. Then
$$
||P(\cdot,t_n)-P_h^n||_{L^2}\leq C(P,\rho \sigma  T,\alpha)(\tau+h^r),~~0<\alpha <1;
$$
and
$$
||P(\cdot,t_n)-P_h^n||_{L^2}\leq C(P,\rho \sigma  T,\alpha)(\tau^\alpha+\tau^{\alpha-1}h^r+h^r),~~\alpha \rightarrow 1,
$$
where $C(P,\rho \sigma  T,\alpha)$ is a constant depending  only on $P,\rho \sigma  T,\alpha$ and  $(\cdot,t_n)\in (0,l) \times (0,T]$, $ n=1,2,\ldots,N$.
\end{theorem}

\begin{proof}
Denoting $\varepsilon^n=P_h^n-R_hP(x,t_{n})$, then from (\ref{3.23}) and Lemma \ref{Lemma3.3}, we get the following error equation
\begin{equation}\label{3.27}
(\varepsilon^n,q_h)+\kappa_\alpha \tau^{\alpha} (\nabla \varepsilon^n,\nabla q_h)
=- \sum_{k=1}^{n-1} d_{k}^\alpha(\varepsilon^{n-k},q_h)+(R_n,q_h),~~ \forall q_h\in S_h(\Omega),
\end{equation}
where $||R_n||_{L^2}=\tau^{\alpha}||r_n||_{L^2}\leq \overline{C}_P \left(\tau^{1+\alpha}+\frac{(\rho \sigma  T)^\alpha+1}{n^\alpha}h^r\right)$.
Taking $q_h=\varepsilon^n$ in (\ref{3.27}), it leads to
$$
||\varepsilon^n||_{L^2} \leq  - \sum_{k=1}^{n-1} d_{k}^\alpha ||\varepsilon^{n-k}||_{L^2}+||R_{n}||_{L^2}
\leq  - \sum_{k=1}^{n-1} g_k ||\varepsilon^{n-k}||_{L^2}+||R_{n}||_{L^2}.
$$
According to Lemma  \ref{lemma2.2}, there exists

(a) when $0<\alpha<1$,
\begin{equation}\label{3.28}
\begin{split}
||R_hP(x,t_{n})-P_h^n||_{L^2}
&=||\varepsilon^n||_{L^2} \leq   \left(\sum_{k=0}^{n-1}g_k\right)^{-1}||R_{n}||_{L^2}
 \leq  n^\alpha  \Gamma(1-\alpha)||R_{n}||_{L^2}\\
& \leq \overline{C}_P  \Gamma(1-\alpha) \Big[T^\alpha\tau+\left((\rho \sigma  T)^\alpha+1\right)h^r\Big];
\end{split}
\end{equation}

(b) when $\alpha \rightarrow1$,
\begin{equation}\label{3.29}
\begin{split}
||R_hP(x,t_{n})-P_h^n||_{L^2}
&=||\varepsilon^n||_{L^2} \leq     n||R_{n}||_{L^2}
\leq \overline{C}_P \left(T\tau^{\alpha}+\frac{(\rho \sigma  T)^\alpha+1}{n^{\alpha-1}}h^r\right)\\
&\leq \overline{C}_P \Big[T\tau^{\alpha}+T^{1-\alpha}((\rho \sigma  T)^\alpha+1)\tau^{\alpha-1}h^r\Big].
\end{split}
\end{equation}
Then, using (\ref{3.26}) and  triangle inequality, it leads to
\begin{equation*}
\begin{split}
||P(x,t_n)-P_h^n||_{L^2}
&\leq ||P(x,t_n)-R_hP(x,t_n)||_{L^2}+||R_hP(x,t_n)-P_h^n||_{L^2}\\
&\leq \widehat{C}_Ph^r +||R_hP(x,t_n)-P_h^n||_{L^2}.
\end{split}
\end{equation*}
Hence, according to (\ref{3.28}), (\ref{3.29}) and above inequality, we obtain
\begin{equation*}
\begin{split}
||P(x,t_n)-P_h^n||_{L^2}
&\leq \widehat{C}_Ph^r+\overline{C}_P  \Gamma(1-\alpha) \Big[T^\alpha\tau+\left((\rho \sigma  T)^\alpha+1\right)h^r\Big]\\
&\leq C(P,\rho \sigma  T,\alpha)(\tau+h^r),~~0<\alpha <1;
\end{split}
\end{equation*}
and
\begin{equation*}
\begin{split}
||P(x,t_n)-P_h^n||_{L^2}
&\leq \widehat{C}_Ph^r+\overline{C}_P \!\Big[T\tau^{\alpha}+T^{1-\alpha}((\rho \sigma  T)^\alpha+1)\tau^{\alpha-1}h^r\Big]\\
&\leq C(P,\rho \sigma  T,\alpha)(\tau^\alpha+\tau^{\alpha-1}h^r+h^r),~\alpha \rightarrow 1.
\end{split}
\end{equation*}

\end{proof}
\begin{remark}
 In particular, when $U(x)=0$, the fractional Feynman-Kac equation (\ref{3.1}) reduces to the celebrated fractional Fokker-Planck equation \cite{Barkai:01,Metzler:00}.
Similarly, the optimal convergent order is also obtained for using finite element method to solve fractional Fokker-Planck equation.
\end{remark}

\section{Numerical Results}
We numerically verify the above theoretical results including convergent
orders and numerical stability.  And the $ l_\infty$ norm is used to measure the numerical errors.
Without loss of generality, we add a force term $f(x,t)$ on the right hand side of (\ref{2.1}), (\ref{1.2}) and  (\ref{3.1}), respectively.
For the numerical schemes, including the first and high order ones, of both forward and backward Feynman-Kac equations, the numerical experiments are also performed to illustrate the validity of the algorithms. In the following we reuse $P(x,\rho,t)$, i.e.,  $P(x,t)$ is replaced by $P(x,\rho,t)$. In fact, by using the algorithm of numerical inversion of Laplace transforms \cite{Abate:95}, we numerically get $P(x,A,t)$; and the marginal PDFs of $A$, $J(A):=\int_{-\infty}^{+\infty} P(x,A,t)dx$, and of $x$, $K(x):=\int_{-\infty}^{+\infty} P(x,A,t)dA$ are also calculated; in particular, the values of $K(x)$ are compared with the ones of $KK(x)$ being the solution of the corresponding fractional Fokker-Planck equation, i.e., the fractional Feynman-Kac equations with $U(x)=0$, to further illustrate the effectiveness of the provided schemes.
\subsection{Numerical results  for   $P(x,\rho,t)$}
\begin{example}[Finite Difference; The forward fractional Feynman-Kac equation (\ref{1.2})]\end{example}
Consider the forward fractional Feynman-Kac equation  (\ref{1.2}), on a finite domain  $0< x < 1 $,  $0<t \leq 1$,  with the coefficient   $\kappa_\alpha=0.5$ and $U(x)=x$, $\rho=1+i$, $i=\sqrt{-1}$, the forcing function
\begin{equation*}
\begin{split}
f(x,t)=
&(3+\alpha)e^{-\rho xt}t^{2+\alpha}\sin(\pi x) \\
&   -\kappa_\alpha\frac{\Gamma(4+\alpha)}{\Gamma(3+2\alpha)}t^{2+2\alpha}e^{-\rho xt}\left(\rho^2 t^2 \sin(\pi x)-2\pi \rho t \cos(\pi x)-\pi^2 \sin(\pi x) \right),
\end{split}
\end{equation*}
the initial condition $P(x,\rho,0)=0 $, and the boundary conditions $P(0,\rho, t)=P(1,\rho,t)=0$. Then (\ref{1.2}) has the exact
solution $$P(x,\rho,t)=e^{-\rho x t}t^{3+\alpha}\sin(\pi x). $$
\begin{table}[h]\fontsize{9.5pt}{12pt}\selectfont
 \begin{center}
  \caption {The maximum errors and convergent orders for  (\ref{2.14}) with  $q=4$, when $U(x)=x$, $\rho=1+i$, $i=\sqrt{-1}$, $\kappa_\alpha=0.5$, $h=\tau^2.$} \vspace{5pt}
\begin{tabular*}{\linewidth}{@{\extracolsep{\fill}}*{10}{c}}                                    \hline  
$\tau$& $\alpha=0.1$&  Rate       & $\alpha=0.5$  & Rate       & $\alpha=0.9$ &   Rate    \\\hline
~~~1/10&  6.4452e-005   &             & 9.3365e-005     &            & 6.7634e-005   &         \\
~~~1/20&  4.2081e-006   &  3.9370     & 5.6941e-006     & 4.0353     & 4.1294e-006   & 4.0337   \\
~~~1/40&  2.8663e-007   &  3.8759     & 3.5337e-007     & 4.0102     & 2.5636e-007   & 4.0097  \\
~~~1/80&  1.5245e-008   &  4.2328     & 2.3108e-008     & 3.9348     & 1.5804e-008   & 4.0198   \\ \hline
    \end{tabular*}\label{table:1}
  \end{center}
\end{table}

Table \ref{table:1}  shows  the maximum errors  at time $T=1$
with $h=\tau^2$; and  the numerical results confirm that the scheme (\ref{2.14}) has the global truncation error $\mathcal{O} (\tau^4+h^2)$.

\begin{example}[Finite Difference; The backward fractional Feynman-Kac equation (\ref{2.1})] \end{example}
Consider the backward fractional Feynman-Kac equation  (\ref{2.1}), on a finite domain  $0< x < 1 $,  $0<t \leq 1$,  with the coefficient  $\kappa_\alpha=0.5$ and $U(x)=x$, $\rho=1+i$, $i=\sqrt{-1}$; the forcing function
\begin{equation*}
\begin{split}
f(x,t)=&\frac{\Gamma(4+\alpha)}{\Gamma(4)}e^{-\rho xt}t^3\sin(\pi x)\\
       &  -\kappa_\alpha e^{-\rho xt}(t^{3+\alpha}+1)\left(\rho^2 t^2 \sin(\pi x)-2\pi \rho t \cos(\pi x)-\pi^2 \sin(\pi x) \right),
\end{split}
\end{equation*}
the initial condition $P(x,\rho,0)=\sin(\pi x) $, and the boundary
conditions $P(0,\rho,t)=P(1,\rho,t)=0$. Then (\ref{2.1}) has the exact
solution $$P(x,\rho,t)=e^{-\rho xt}(t^{3+\alpha}+1)\sin(\pi x). $$
\begin{table}[h]\fontsize{9.5pt}{12pt}\selectfont
 \begin{center}
  \caption {The maximum errors and convergent orders for (\ref{2.10}), i.e., $q=1$, when $U(x)=x$, $\rho=1+i$, $i=\sqrt{-1}$, $\kappa_\alpha=0.5$, $\tau=h^2.$} \vspace{5pt}
\begin{tabular*}{\linewidth}{@{\extracolsep{\fill}}*{10}{c}}                                    \hline  
$h $& $\alpha=0.1$&  Rate       & $\alpha=0.5$  & Rate       & $\alpha=0.9$ &   Rate    \\\hline
~~~1/10&  1.1562e-002  &             & 1.1953e-002    &            & 1.3563e-002   &         \\
~~~1/20&  2.9178e-003  &  1.9864    & 3.0170e-003    & 1.9861    & 3.4222e-003  & 1.9867   \\
~~~1/40&  7.3118e-004  &  1.9966    & 7.5624e-004     & 1.9962     &8.5874e-004   & 1.9946   \\
~~~1/80&  1.8290e-004  &  1.9991     & 1.8930e-004     & 1.9981     &  2.1484e-004   & 1.9990  \\ \hline
    \end{tabular*}\label{table:2}
  \end{center}
\end{table}

\begin{table}[h]\fontsize{9.5pt}{12pt}\selectfont
 \begin{center}
  \caption {The maximum errors and convergent orders for  (\ref{2.12}) with  $q=4$, when $U(x)=x$, $\rho=1+i$, $i=\sqrt{-1}$, $\kappa_\alpha=0.5$, $h=\tau^2.$} \vspace{5pt}
\begin{tabular*}{\linewidth}{@{\extracolsep{\fill}}*{10}{c}}                                    \hline  
$\tau$& $\alpha=0.1$&  Rate       & $\alpha=0.5$  & Rate       & $\alpha=0.9$ &   Rate    \\\hline
~~~1/10&  1.1563e-004  &             & 1.0990e-004    &            & 1.0487e-004   &         \\
~~~1/20&  7.2278e-006  &  3.9998     & 6.8693e-006     & 3.9998     & 6.4844e-006   & 4.0154   \\
~~~1/40&  4.5173e-007  &  4.0000     & 4.2926e-007     & 4.0003     & 4.0409e-007   & 4.0042  \\
~~~1/80&  2.7781e-008  &  4.0233     & 2.6565e-008     & 4.0143     &  2.5510e-008  & 3.9855   \\ \hline
    \end{tabular*}\label{table:3}
  \end{center}
\end{table}

Table \ref{table:2} and  Table \ref{table:3} show that the algorithms with $q=1$ and $4$ have the global truncation errors $\mathcal{O} (\tau+h^2)$ and $\mathcal{O} (\tau^4+h^2)$ at time $T=1$, respectively.

\begin{example}[Finite Element; The backward fractional Feynman-Kac equation (\ref{3.1})]\end{example}
We use the finite element method (\ref{3.23}) with the piecewise linear polynomial approximation ($r=2$) in space to solve
the backward fractional Feynman-Kac equation (\ref{3.1}), on a finite domain  $0< x < 1 $,  $0<t \leq 1$,  with the coefficient
$\kappa_\alpha=0.5$, $U(x)=1$, $\rho=1+ i $, $i=\sqrt{-1}$,  the forcing function
\begin{equation*}
\begin{split}
f(x,t)=& \frac{\Gamma(4+\alpha)}{\Gamma(4)}   e^{-\rho t}t^3\sin(\pi x)
        +\kappa_\alpha\pi^2 e^{-\rho t}(t^{3+\alpha}+1)\sin(\pi x) ,
\end{split}
\end{equation*}
the initial condition $P(x,\rho,0)=\sin(\pi x)$,  and the boundary
conditions $P(0,\rho,t)=P(1,\rho,t)=0$. Then (\ref{3.1}) has the exact
solution $$P(x,\rho,t)= e^{-\rho t}(t^{3+\alpha}+1)\sin(\pi x). $$

\begin{table}[h]\fontsize{9.5pt}{12pt}\selectfont
 \begin{center}
  \caption {The maximum errors and convergent orders for the finite element method (\ref{3.23}), when $U(x)=1$, $\rho=1+i$, $i=\sqrt{-1}$, $\kappa_\alpha=0.5$, $\tau=h^2.$} \vspace{5pt}
\begin{tabular*}{\linewidth}{@{\extracolsep{\fill}}*{10}{c}}                                    \hline  
$h$& $\alpha=0.1$&  Rate       & $\alpha=0.5$  & Rate       & $\alpha=0.9$ &   Rate    \\\hline
~~~1/10&  6.5499e-004  &             & 1.6544e-003    &            & 3.6377e-003   &         \\
~~~1/20&  1.6658e-004  &  1.9753     & 4.2039e-004    & 1.9765     & 9.2480e-004   & 1.9758   \\
~~~1/40&  4.1822e-005  &  1.9939     & 1.0552e-004    & 1.9942    & 2.3216e-004   & 1.9940  \\
~~~1/80&  1.0467e-005  &  1.9985     & 2.6407e-005     & 1.9985     & 5.8101e-005   & 1.9985   \\ \hline
    \end{tabular*}\label{table:4}
  \end{center}
\end{table}

Table \ref{table:4} shows the maximum errors at time $T=1$
with $\tau=h^2$, and  the numerical results confirm that the finite element method has the global truncation error $\mathcal{O} (\tau+h^2)$.

\subsection{Simulations with Dirac delta function as initial condition}
Let the joint probability density function $P(x,A,t)$ be a real function of $A$, with $P(x,A,t)=0$ for $A<0$; the Laplace transform and its inversion formula are defined as follows:
\begin{equation}\label{4.1}
\begin{split}
& P(x,\rho ,t)=\mathcal{L}\{ P(x,A,t)\}=\int_0^\infty P(x,A,t)e^{-\rho A}dA,~~~~\rho=\nu+i\omega; \\
& P(x,A ,t)=\mathcal{L}^{-1}\{ P(x,\rho,t)\}=  \frac{1}{2\pi i}\int_{\nu-i\infty}^{\nu+i\infty} P(x,\rho,t)e^{\rho A}d\rho, \\
\end{split}
\end{equation}
where $\nu>0$ is arbitrary, but is greater than the real parts of all the singularities  of $P(x,\rho,t)$.
According to Abate's method \cite{Abate:95} (or see Appendix), we can take
 \begin{equation}\label{4.2}
   \rho_k=\nu+i \frac{\pi}{A}k,~~\mbox {with}~~ \nu=\frac{18.4}{2A},~~ i=\sqrt{-1},~~k=0,1,2,\ldots.
 \end{equation}

Simulate the forward and backward fractional Feynman-Kac equations (\ref{2.14}) and  (\ref{2.12}), respectively, on a finite domain  $0< x < 1 $,  $0<t \leq T$,  with the coefficient   $\kappa_\alpha=0.5$ and
the forcing function $f(x,t)=0$ and take
\begin{equation}\label{4.3}
~~~~~U(x)=\left\{ \begin{array}
{l@{\quad} l}1,~~~~x \in (0.5,1),\\
0,~~~~{\rm otherwise};
\end{array}
\right.
\end{equation}

\begin{equation}\label{4.4}
~~{\rm or}~~~~~~~U(x)=\left\{ \begin{array}
{l@{\quad} l}1,~~~~x \in (0.25,0.75),\\
  0,~~~~{\rm otherwise}.
\end{array}
\right.
\end{equation}
The initial condition $P(x,\rho,0)=\delta_a(x-0.5)$ (Dirac delta function), and the boundary conditions $P(0,\rho,t)=P(1,\rho,t)=0$,
where $\rho=\{\rho_k\}_{k=0}^{35}$. The Dirac delta function is defined by the limit of the sequence of Gaussians
\begin{equation*}
  \delta_a(x)=\frac{1}{2\sqrt{a\pi}}e^{-\frac{x^2}{4a}}~~~~{\rm as}~~~~a\rightarrow 0;
\end{equation*}
and we take $a=0.0005$ as the approximation in numerical computations.

\vskip 0.2cm
The corresponding  procedure of generating Figures \ref{FIG.1}-\ref{FIG.6} is executed as follows:
\begin{description}
\item[(1)] For every fixed $\alpha$, $A$ and $\rho_k$ (k=0,1,\ldots,35), according to   (\ref{2.12}) or (\ref{2.14}),
           we obtain $P(x,\rho_k,t)$ at time $T$ with  $q=2$, $\tau=h=1/500$ in (\ref{2.12}) or (\ref{2.14}).
\item[(2)] From Abate's method \cite{Abate:95} (see the Appendix), we get $P(x,A,t)$.
\item[(3)] According to the composite trapezoidal formula, we get $J(A)=\int_0^1 P(x,A,t)dx$.
\end{description}

\begin{figure}[t]
    \begin{minipage}[t]{0.50\linewidth}
    \includegraphics[scale=0.45]{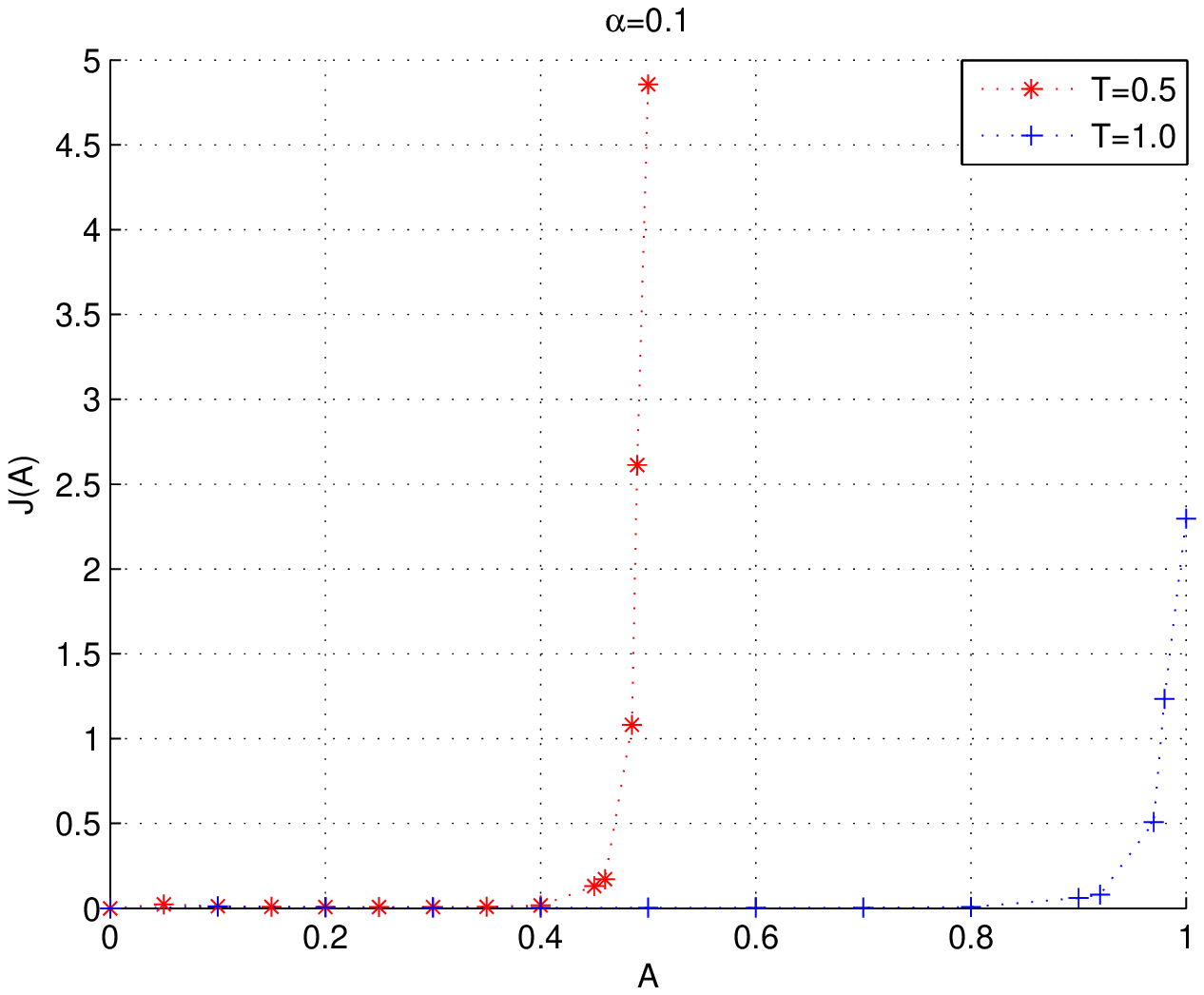}
    \caption{$J(A)$ for (\ref{2.12}) with $U(x)$ defined in (\ref{4.3}).}  \label{FIG.1}
    \end{minipage}
    \begin{minipage}[t]{0.50\linewidth}
    \includegraphics[scale=0.45]{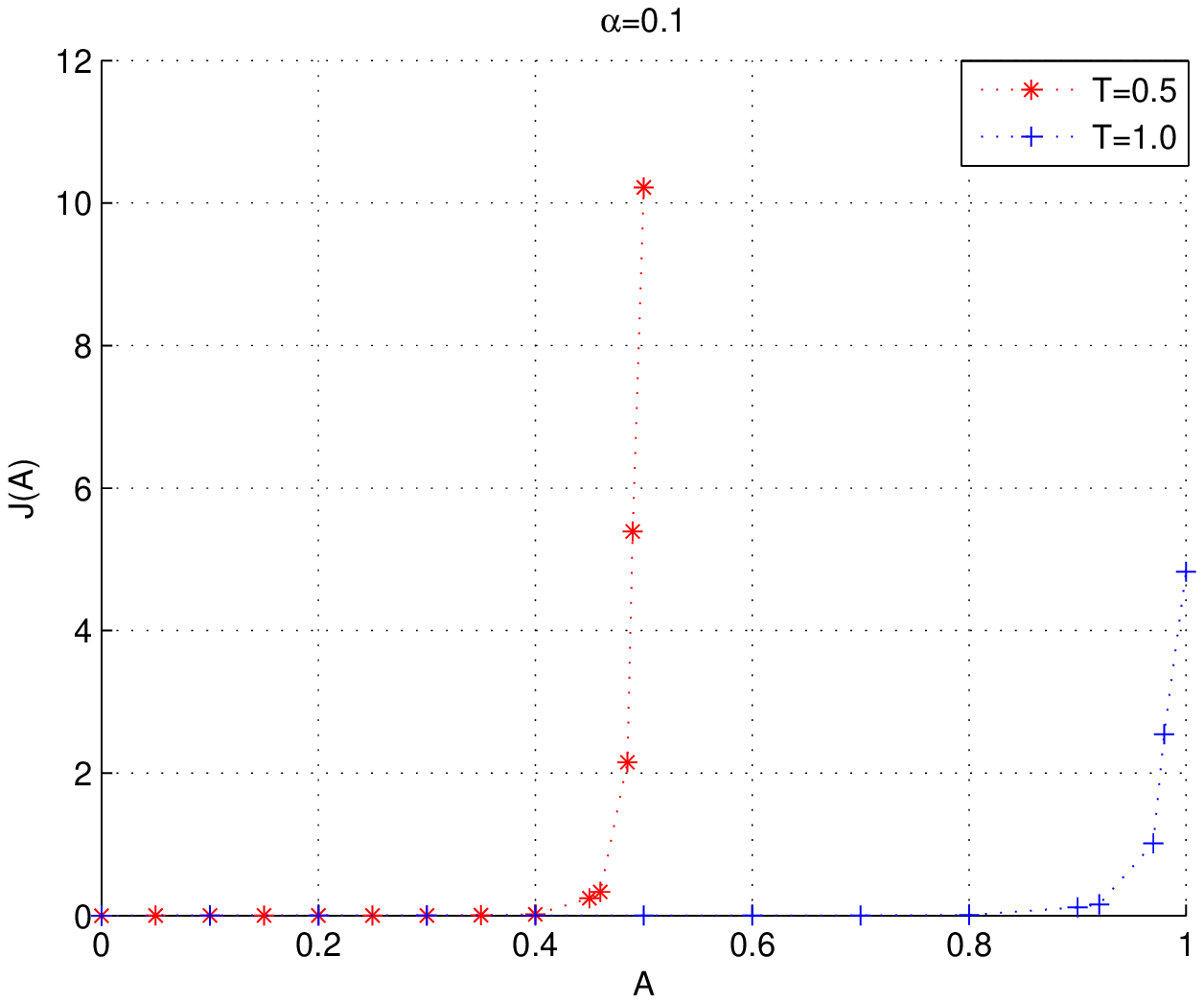}
    \caption{$J(A)$ for (\ref{2.12}) with $U(x)$ defined in (\ref{4.4}).}  \label{FIG.2}
    \end{minipage}
    \begin{minipage}[t]{0.50\linewidth}
    \includegraphics[scale=0.45]{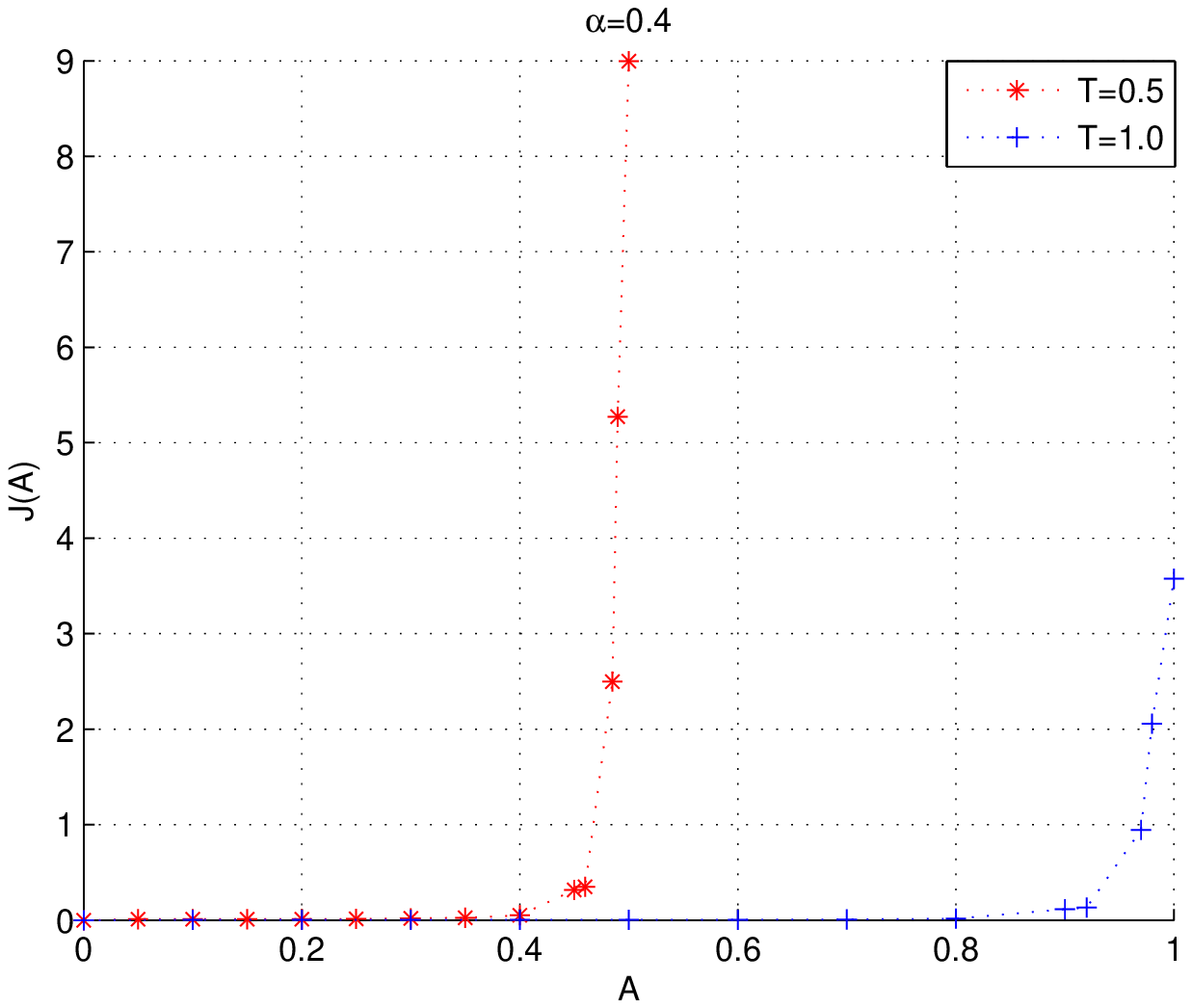}
    \caption{$J(A)$ for (\ref{2.12}) with $U(x)$ defined in (\ref{4.4}).}  \label{FIG.3}
    \end{minipage}
    \begin{minipage}[t]{0.50\linewidth}
    \includegraphics[scale=0.45]{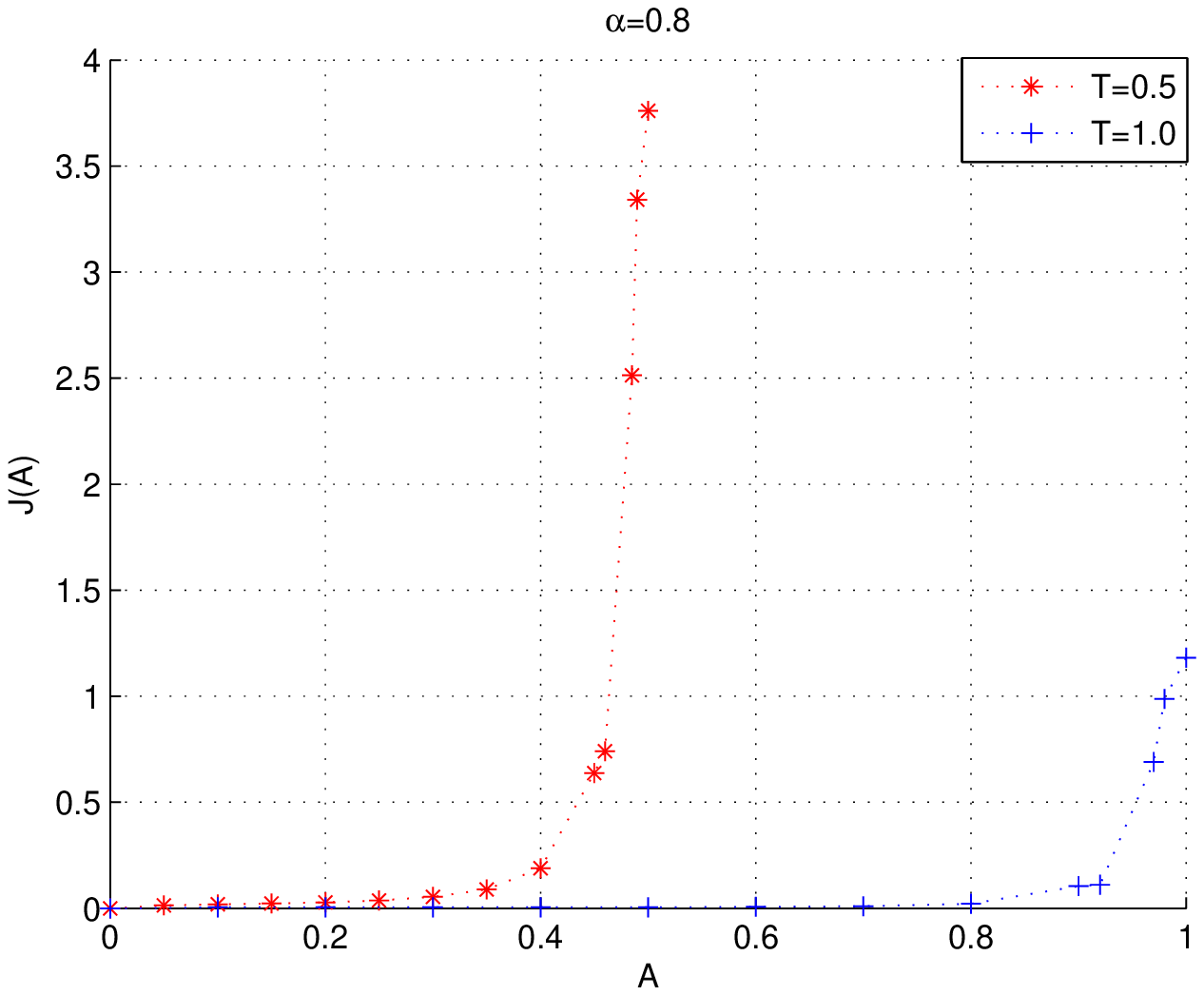}
    \caption{$J(A)$ for (\ref{2.12}) with $U(x)$ defined in (\ref{4.4}).}  \label{FIG.4}
    \end{minipage}
    \begin{minipage}[t]{0.50\linewidth}
    \includegraphics[scale=0.45]{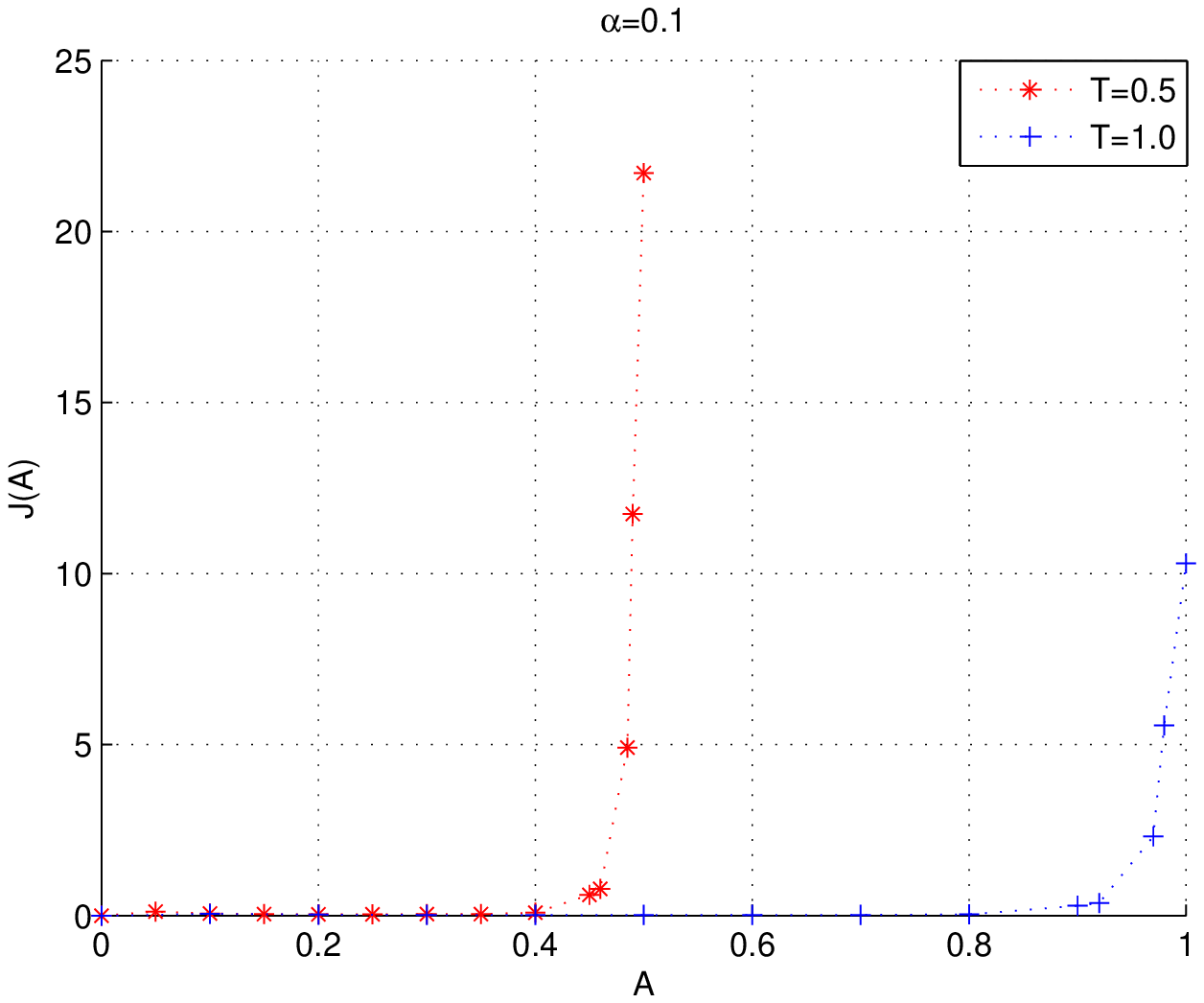}
    \caption{$J(A)$ for (\ref{2.14}) with $U(x)$ defined in (\ref{4.3}).}  \label{FIG.5}
    \end{minipage}
    \begin{minipage}[t]{0.50\linewidth}
    \includegraphics[scale=0.45]{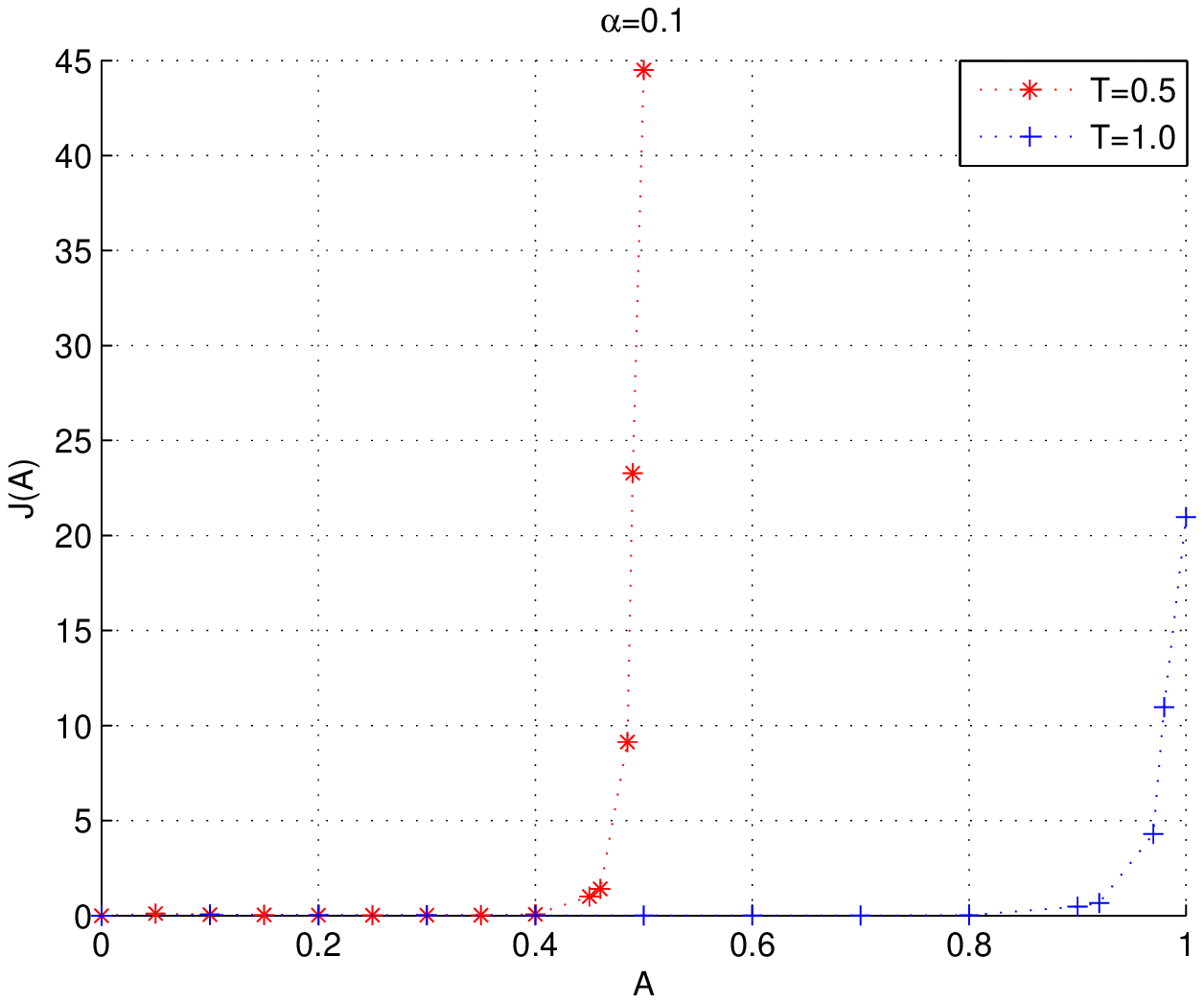}
    \caption{$J(A)$ for (\ref{2.14}) with $U(x)$ defined in (\ref{4.4}).}  \label{FIG.6}
    \end{minipage}
\end{figure}

Figures \ref{FIG.1}-\ref{FIG.6} show that the conservation of probability, i.e., the areas under the curves at time $T=0.5$ and $T=1.0$ are almost the same.

\vskip 0.2cm
The corresponding procedure of generating Figures \ref{FIG.7}-\ref{FIG.8} is executed as follows:

\begin{description}
\item[(1)] For every fixed $\alpha$, $A$ and $\rho_k$ (k=0,1,\ldots,35), according to   (\ref{2.12}) or (\ref{2.14}),
           we obtain $P(x,\rho_k,t)$ at time $T$ with  $q=2$, $\tau=h=1/500$ in (\ref{2.12}) or (\ref{2.14}).
\item[(2)] From Abate's method \cite{Abate:95} (see the Appendix), we get $P(x,A,t)$.
\item[(3)] Using the composite trapezoidal formula, we get $K(x):=P_{FFK}(x,t)=\int_0^T P(x,A,t)dA$ with $U(x)=1$ on the domain $x\in (0,1)$, where $A=(0:0.008:0.2)$ is used in calculation.
\item[(4)] When $U(x)=0$ on the domain $x\in (0,1)$, using (\ref{2.12}) or (\ref{2.14}), we obtain $KK(x):=P_{FFP}(x,t)$.
\end{description}

\begin{figure}[t]
    \begin{minipage}[t]{0.50\linewidth}
    \includegraphics[scale=0.45]{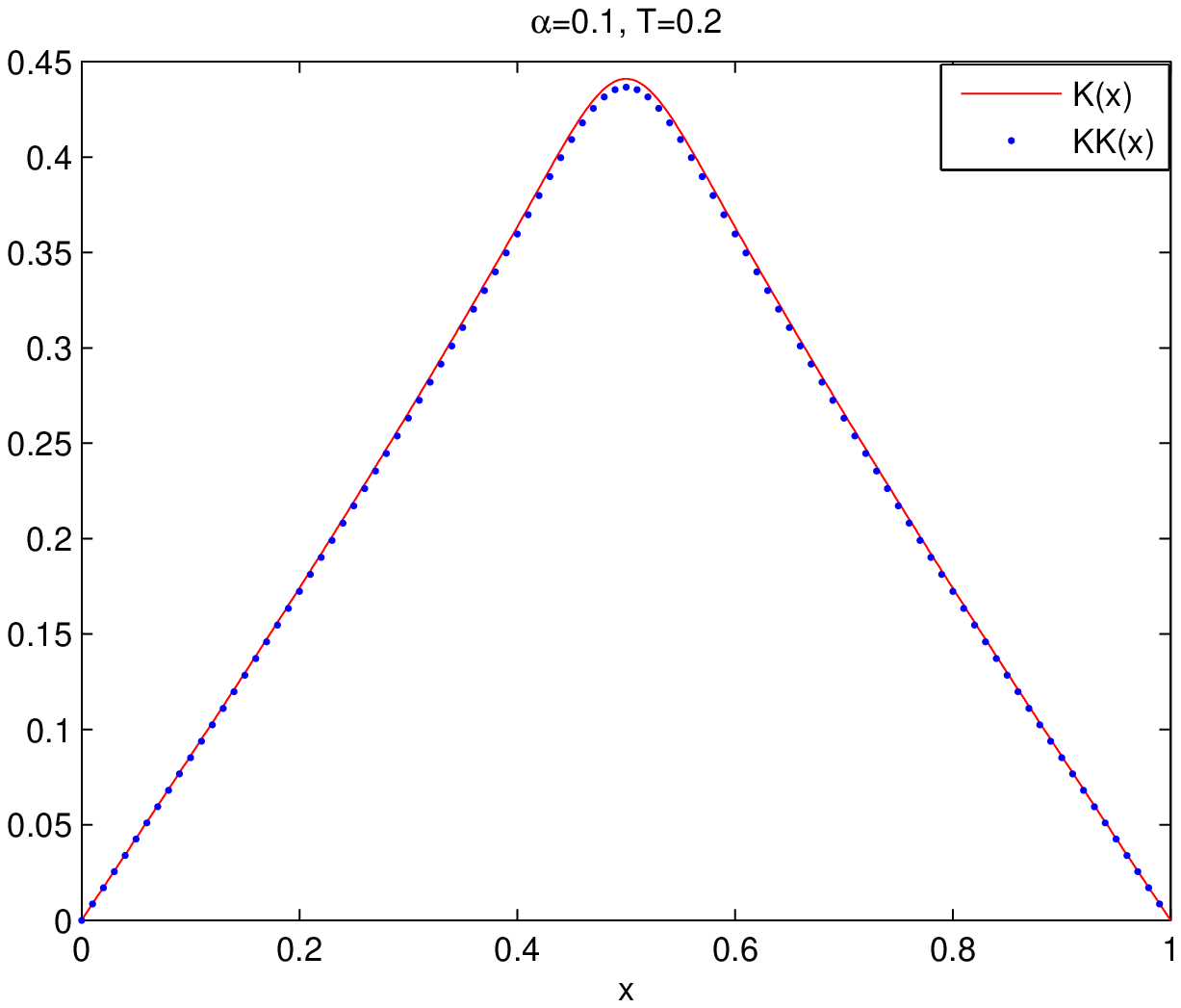}
    \caption{$K(x)$ and $KK(x)$ for (\ref{2.12}), respectively.}  \label{FIG.7}
    \end{minipage}
    \begin{minipage}[t]{0.50\linewidth}
    \includegraphics[scale=0.45]{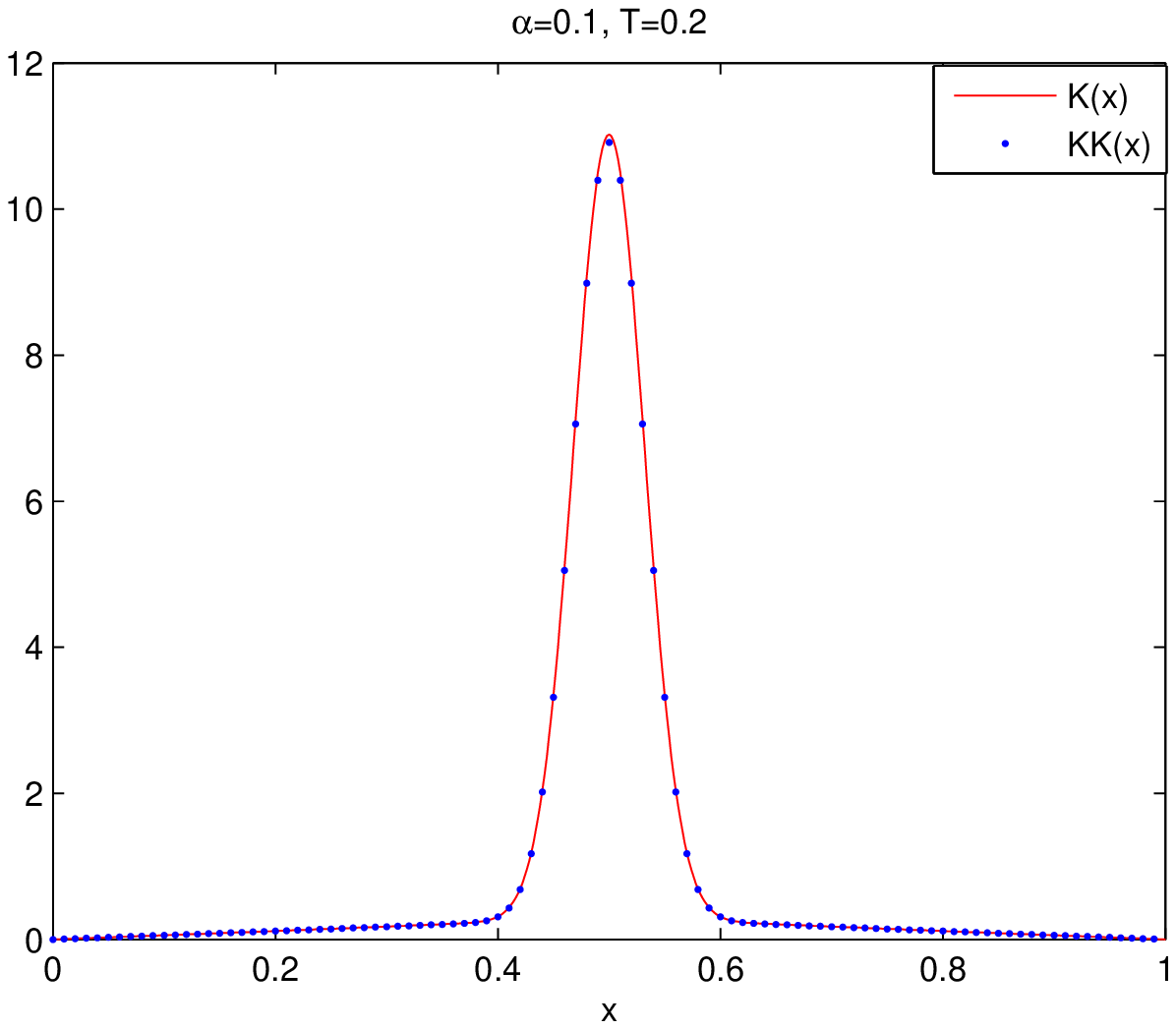}
    \caption{$K(x)$ and $KK(x)$ for (\ref{2.14}), respectively.}  \label{FIG.8}
    \end{minipage}
\end{figure}

Figures \ref{FIG.7}-\ref{FIG.8} show that $K(x)=KK(x)$, which further confirms the effectiveness of the provided schemes and, of course, the correctness of the algorithm of the numerical inversion of Laplace transformations.
\section{Conclusion}

The properties and numerical discretizations of fractional substantial derivative are detailedly analyzed in  \cite{Chen:13}.  This paper further discusses the numerical algorithms for the forward and backward fractional Feynman-Kac equations with fractional substantial derivative. Finite difference methods are used to solve both the forward and backward fractional Feynmann-Kac equations; and the finite element methods are applied to solve the backward fractional Feynmann-Kac equation. The finite difference scheme with  first order accuracy in time direction and the finite element methods for the backward Feynmann-Kac equation are theoretically analyzed, including the unconditional stability and the convergence; in particular,  the optimal convergent order is obtained for the finite element method. Extensive numerical experiments are performed for the schemes of both forward and backward fractional Feynmann-Kac equations. Especially, when $U(x)=0$, both the forward and backward fractional Feynman-Kac equations reduce to the celebrated fractional Fokker-Planck equation. By comparing the marginal PDF of the solutions of both forward and backward fractional Feynmann-Kac equations with the solution of fractional Fokker-Planck equation, the effectiveness of the proposed schemes are further verified.

\section*{Acknowledgments} This work was supported by the National Natural Science Foundation of China under
Grant No. 11271173.

\section*{Appendix}

To prove that ({\ref{1.4}}) is equivalent to ({\ref{1.5}}), we first introduce some properties of the fractional substantial calculus.

For $0<\alpha<1$, using  Lemma 2.3, Lemma 2.7 and Lemma 2.6 of \cite{Chen:13}, there exists
\begin{gather*}
{^s\!}D_t^\alpha[{^s\!}I_t^\alpha P(t)] =P(t); \tag{$A.1$}\\
{^s\!}I_t^\alpha[{^s\!}D_t^\alpha P(t)] =P(t)-[{^s\!}D_t^{\alpha-1}P(t)]_{t=0}\frac{ t^{\alpha-1} e^{-\rho U(x) t}  }{\Gamma(\alpha)};\tag{$A.2$}\\
{^s_c}{D}_t^\alpha P(x,t)= {^s\!}D_t^{\alpha-1}[{^s\!}D_t P(t)]
={^s\!}D_t^{\alpha}P(t)-\frac{ t^{-\alpha} e^{-\rho U(x) t}  }{\Gamma(1-\alpha)}P(0)\tag{$A.3$}.
\end{gather*}


\noindent{\bf Theorem A.1} \emph{  Let  $0<\alpha<1$ and $P(x,t) \in C_{x,t}^{2,1}[a,b]\times[0,T]$. Then
  \begin{equation}
  {^s\!}D_tP(x,t)= {^s\!}D_t^{1-\alpha}\left[\kappa_\alpha \frac{\partial^2}{\partial x^2} P(x,t) \right],\tag{$A.4$}
\end{equation}
is equivalent to
\begin{equation*}
{^s_c}{D}_t^\alpha P(x,t)={^s\!}D_t^\alpha P(x,t)-\frac{t^{-\alpha}e^{-\rho U(x) t}}{\Gamma(1-\alpha)}    P(x,0)
  = \kappa_\alpha \frac{\partial^2}{\partial x^2} P(x,t). \tag{$A.5$}
\end{equation*} }
\begin{proof}
Derive $(A.5)$ from $(A.4)$. Performing both sides of $(A.4)$ by ${^s\!}D_t^{\alpha-1}$ leads to
  \begin{equation*}
{^s\!}D_t^{\alpha-1} [{^s\!}D_tP(x,t)]
    ={^s\!}D_t^{\alpha-1}\left\{ {^s\!}D_t^{1-\alpha}\left[\kappa_\alpha \frac{\partial^2}{\partial x^2} P(x,t) \right]\right\}.
\end{equation*}
According to above equation and  $(A.3)$, $(A.2)$, we get
\begin{equation*}
{^s\!}D_t^{\alpha}P(x,t)-\frac{ t^{-\alpha} e^{-\rho U(x) t}P(x,0)  }{\Gamma(1-\alpha)}
=\kappa_\alpha \frac{\partial^2}{\partial x^2} P(x,t)
  -{^s\!}D_t^{-\alpha}\left[\kappa_\alpha \frac{\partial^2}{\partial x^2} P(x,t)\right]_{t=0}\frac{ t^{-\alpha} e^{-\rho U(x) t}  }{\Gamma(1-\alpha)}.
\end{equation*}
If a function $P(t)$ is continuously differentiable in the closed interval $[0,t]$, then
\begin{equation*}
\begin{split}
&{^s\!}D_t^{-\alpha}P(t)|_{t=0}\\
&\quad =\lim_{ t\rightarrow 0^{+}}\frac{1}{\Gamma(\alpha)}\int_{0}^t{\left(t-\tau\right)^{\alpha-1}}e^{-\rho U(x)(t-\tau)}{P(\tau)}d\tau\\
&\quad=\lim_{ t\rightarrow 0^{+}}\left[ \frac{t^\alpha e^{-\rho U(x) t}}{\Gamma(\alpha+1)} P(0)
+\frac{1}{\Gamma(\alpha+1)}\int_{0}^t{\left(t-\tau\right)^{\alpha}}e^{-\rho U(x)(t-\tau)}{[\rho U(x) P(\tau)+P'(\tau)]}d\tau \right]\\
&\quad=0.
\end{split}
\end{equation*}
Since  $P(x,t) \in C_{x,t}^{2,1}[a,b]\times[0,T]$, we have
$${^s\!}D_t^{-\alpha}\left[\kappa_\alpha \frac{\partial^2}{\partial x^2} P(x,t)\right]_{t=0}=0.$$
It implies that  $(A.5)$ holds.

Derive $(A.4)$ from $(A.5)$. Performing both sides of $(A.5)$ by ${^s\!}D_t^{1-\alpha}$ results in
\begin{equation*}
{^s\!}D_t^{1-\alpha}\left[{^s\!}D_t^\alpha P(x,t)-\frac{t^{-\alpha}e^{-\rho U(x) t}}{\Gamma(1-\alpha)}    P(x,0)\right]
  ={^s\!}D_t^{1-\alpha}\left[ \kappa_\alpha \frac{\partial^2}{\partial x^2} P(x,t)\right].
\end{equation*}
Using $(A.3)$ and $(A.1)$, there exists
 \begin{equation*}
{^s\!}D_t^{1-\alpha}\left[{^s\!}D_t^\alpha P(x,t)-\frac{t^{-\alpha}e^{-\rho U(x) t}}{\Gamma(1-\alpha)}    P(x,0)\right]
  ={^s\!}D_t^{1-\alpha}\left[ {^s\!}D_t^{\alpha-1}[{^s\!}D_t P(x,t)]\right]
  = {^s\!}D_t P(x,t).
\end{equation*}
That is  $(A.4)$ holds.

\end{proof}

\begin{algorithm*}
\% For the convenience to the reader, we add Matlab codes for the inverse Laplace transforms 
\%  used in this paper; for the details of the derivation of the algorithm, one can refer to \cite{Abate:95}.
\caption{The MATLAB Program for Numerical Inversion of Laplace Transforms }
\begin{algorithmic}[1]
\STATE function [Uappr,error]=Laplace Feynman Euler $(A,x)$
\STATE $m=20$;  $Ntr=15$; $U=\exp(18.4/2)/A$; $X=18.4/(2*A)$; $H=\pi/A$;
\STATE C(1)=1;
\STATE for i=2:m+1
\STATE ~~~~C(i)=C(i-1)*(m+1-i)/(i-1);
\STATE end
\STATE \% Assume that $P(x,\rho_k,t)$ is obtained with $k=0,1,\ldots, m+Ntr$ for fixed $x$, $t$ and $A$.
\STATE Sum= $P(x,\rho_0,t)/2$;
\STATE fnRf(k)=real$\left(P(x,\rho_k,t)\right)$;
\STATE for N=1:Ntr
\STATE ~~~~$Y=N*H$;
\STATE ~~~~Sum=Sum+$(-1)^N*$fnRf(N);
\STATE end
\STATE SU(1)=Sum;
\STATE for K=1:m
\STATE ~~~~N=Ntr+K;
\STATE ~~~~Y=N*H;
\STATE ~~~~SU(K+1)=SU(K)+$(-1)^N*$fnRf(N);
\STATE end
\STATE Avgsu=0; Avgsu1=0;
\STATE for J=1:m
\STATE ~~~~Avgsu=Avgsu+C(J)*SU(J);
\STATE ~~~~Avgsu1=Avgsu1+C(J)*SU(J+1);
\STATE end
\STATE Uappr=U*Avgsu$/2^{(m-1)}$; Uappr1=U*Avgsu1/$2^{(m-1)}$;
\STATE error=abs(Uappr-Uappr1)/2.
\end{algorithmic}
\end{algorithm*}

\end{document}